%% file: main.tex
\newif\iftwocolumn
\newif\ifonecolumn
\newif\iflncs
\newif\ifsp
\newif\ifanon
\newif\ifnonanon
\newif\iflong
\newif\ifshort
\newif\ifappendices

\appendicestrue
\lncstrue
\nonanontrue

\iflncs
\onecolumntrue
\longtrue
\documentclass[runningheads]{llncs}
\pdfoutput=1
\fi

\ifsp

\fi

\usepackage{preamble}

\input{macros}

\begin{document}
\captionsetup[figure]{font=small,skip=4pt}
\captionsetup[algorithm]{font=small,skip=4pt}

\input{title}

\maketitle

\input{abstract}

\section{Introduction}
\label{sec:introduction}
\input{introduction}

\section{Preliminaries \& Model}
\label{sec:model}
\input{model}

\section{Construction}
\label{sec:protocol}
\input{protocol}

\section{System Considerations}
\label{sec:implementation}
\input{implementation}

\section{Analysis}
\label{sec:analysis}
\input{analysis-simple}

\ifnonanon
\section*{Acknowledgements}
We thank Shresth Agrawal, Kostis Karantias, Angel Leon, Joachim Neu, and Apostolos Tzinas
for several insightful discussions on this project.
Ertem Nusret Tas is supported by the Stanford Center for Blockchain Research. Lei Yang is supported by a gift from the Ethereum Foundation.
\fi

\bibliographystyle{splncs04}
\bibliography{references}

\ifappendices
\ifsp
  \appendix
\fi
\iflncs
  \appendix
\fi

\section{Attack on SPV Clients on Lazy Blockchains}
\label{sec:attack}
\input{appendix_attack}

\section{Bisection Game}
\label{sec:bisection-game}
\input{appendix_bisection_game}

\section{Latency-Bandwidth Trade-off for Bisection Games}
\label{sec:latency-bandwidth-tradeoff}
\input{appendix_latency_tradeoff}

\section{Superlight Clients}
\label{sec:super_light_clients}
\input{super_light_clients}

\section{Generalizing the Model}
\label{sec:generalize}
\input{generalize}

\section{\COracle Constructions}
\label{sec:coracle-constructions}
\input{appendix_consensus_oracles}

\section{Proofs}
\label{sec:proofs}
\input{appendix_proofs}
\fi

\end{document}

%% file: macros.tex
\newcommand{\LOGdirty}[2]{\ensuremath{\mathbb{L}_{#1}^{#2}}}
\newcommand{\ledgercup}{\ensuremath{\mathbb{L}}^{\cup}}
\newcommand{\ledgercap}{\ensuremath{\mathbb{L}}^{\cap}}
\newcommand{\ledger}{\ensuremath{\mathbb{L}}}
\newcommand{\ledgeraug}{\ensuremath{\mathbb{L}_{+}}}

\newcommand{\largest}{\ensuremath{\overline{\prover}}}

\newcommand{\dtreesp}[0]{\ensuremath{\mathcal{T}}}
\newcommand{\mroot}{\ensuremath{\left<\dtreesp\right>}}

\newcommand{\UTXOP}{\ensuremath{\mathsf{UTXO}^p}}
\newcommand{\UTXOL}{\ensuremath{\mathsf{UTXO}^l}}

\newcommand{\prover}{\ensuremath{\mathcal{P}}}

\newcommand{\verifier}{\ensuremath{\mathcal{V}}}

\newcommand{\eoracle}{execution oracle\xspace}
\newcommand{\Eoracle}{Execution oracle\xspace}
\newcommand{\EOracle}{Execution Oracle\xspace}
\newcommand{\coracle}{consensus oracle\xspace}
\newcommand{\Coracle}{Consensus oracle\xspace}
\newcommand{\COracle}{Consensus Oracle\xspace}

\newcommand{\wf}{well-formed\xspace}
\newcommand{\Wf}{Well-formed\xspace}
\newcommand{\wfc}{well-formedness\xspace}

\newcommand{\transition}{\delta}

\newcommand{\genesisstate}{st_0}
\newcommand{\genesisstatec}{\left<st_0\right>}
\newcommand{\st}{\ensuremath{\mathsf{st}}}
\newcommand{\stc}{\ensuremath{\left<\mathsf{st}\right>}}

\def\chain{\mathcal{C}}

\newcommand{\ie}[0]{\emph{i.e.}\xspace}
\newcommand{\eg}[0]{\emph{e.g.}\xspace}
\newcommand{\cf}[0]{\emph{cf.}\xspace}

\newcommand{\tx}[0]{\ensuremath{\mathsf{tx}}}

\newcommand{\concat}{\,\|\,}

\newcommand{\challenger}[1]{{\color{blue} #1}}
\newcommand{\responder}[1]{{\color{orange} #1}}
\newcommand{\negl}{\mathrm{negl}}

%% file: title.tex
\title{Light Clients for Lazy Blockchains}%
\blfootnote{Authors are listed alphabetically. Contact author: DT.\\Full version of the paper with appendices is available at \href{https://eprint.iacr.org/2022/384}{eprint.iacr.org/2022/384}~\cite{full-version}.}%

\iflncs
\author{%
Ertem Nusret Tas\inst{1} \and%
David Tse\inst{1} \and%
Lei Yang\inst{2} \and%
Dionysis Zindros\inst{1}}%
\institute{
Stanford University\\
\email{\{nusret,dntse,dionyziz\}@stanford.edu}\\
\and
MIT CSAIL\\
\email{leiy@csail.mit.edu}
}

%% file: abstract.tex
\begin{abstract}
\emph{Lazy} blockchains decouple consensus from transaction verification and execution to increase throughput.
Although they can contain invalid transactions (\eg, double spends) as a result, these can easily be filtered out by full nodes that check if there have been previous conflicting transactions. However, creating light (SPV) clients that do not see the whole transaction history becomes a challenge:
A record of a transaction on the chain does not necessarily entail transaction confirmation.
In this paper, we devise a protocol that enables the creation of efficient light clients for lazy blockchains.
The number of interaction rounds and the communication complexity of our protocol are logarithmic in the blockchain execution time.
Our construction is based on a bisection game that traverses the Merkle tree containing the ledger of all -- valid or invalid -- transactions.
We prove that our proof system is succinct, complete and sound, 
and empirically demonstrate the feasibility of our scheme.
\end{abstract}

%% file: introduction.tex
A traveler in Naples saw twelve beggars lying in the sun.
He offered a lira to the laziest of them.
Eleven of them jumped up to claim it, so he gave it to the twelfth~\cite{idleness}.
Towards scalable blockchains,
the holy grail of cryptocurrency adoption,
it has become clear that \emph{lazy} systems
will similarly win the race.

\emph{Eager} blockchain protocols, such as Bitcoin and Ethereum, combine transaction verification and execution with consensus
to ensure that only \emph{valid} transactions are included in their ledger.
In contrast,
lazy blockchain protocols
separate the consensus layer
(responsible for ordering transactions)
from the execution layer
(responsible for interpreting them) to remove the execution bottleneck on scalability~\cite{lazyledger}.
In these systems, a population of \emph{consensus nodes}
collects transactions and places them in a total order,
without care for their validity.
This produces a confirmed \emph{dirty ledger},
a sequence of totally ordered, but potentially invalid,
transactions -- such as double spends.
It is the responsibility of \emph{full nodes} to
\emph{sanitize} the dirty ledger and ascertain which
transactions are valid.
This is done by executing the valid transactions one by one, and ignoring transactions
that are not applicable.
Examples of lazy distributed ledger protocols include Celestia (LazyLedger)~\cite{lazyledger},
Prism~\cite{prism}, Parallel Chains~\cite{parallel-chains} and Snap-and-Chat~\cite{snap-and-chat}.

As consensus nodes do not execute transactions, 
they also cannot find and include the state commitments in the blockchain.
\begin{figure}[t]
    \centering
    \includegraphics[width=0.75\columnwidth]{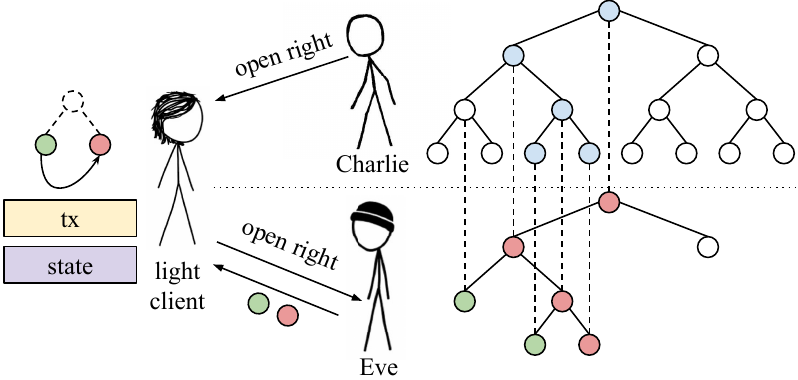}
    \caption{Bisection Game. Charlie the challenger helps the
    light client iteratively traverse the tree of Eve the evil responder.
    A green node indicates a match, while a red node indicates a mismatch, between the two dirty trees.}
    \label{fig:bisection_game}
\end{figure}
Hence, lazy protocols cannot easily support light clients,
and techniques from the realm of eager systems,
such as SPV (Simple Payment Verification~\cite{bitcoin,ethereum_spv}),
are not applicable (\cf \ifappendices Appendix~\ref{sec:attack} \else\cite[Appendix A]{full-version} \fi
for an attack on the succinctness of the SPV clients on lazy blockchains).
For instance, in Ethereum, to prove to the light clients their current account balance, a full node presents a block header containing a state commitment, together with a Merkle inclusion proof of the account to be verified within the commitment.
However, in a lazy protocol, commitments posted to the blockchain cannot be trusted since the consensus nodes do not check the validity of the state.

In this paper, we resolve this outstanding problem by introducing the
first \emph{light client for lazy blockchains}.
Consider a light client, such as the mobile wallet of a vendor, wishing to confirm
an incoming payment.
Our construction allows it
to synchronize with the network and quickly learn its latest account balance.
Towards this purpose, the light client
first connects to several full nodes (\eg, servers by Infura, Chainlink, Alchemy), at least one of which is honest (existential honesty).
It then asks them its account balance.
If the answers received contradict each other, then it
deduces that at least one of them is adversarial.
It interactively interrogates the full nodes
in order to determine which of them is truthful.

Lazy blockchains also appear in the context of optimistic rollups on Ethereum \cite{arbitrum,optimism,fuel}.
In these rollups, transactions are bundled and posted to Ethereum.
Then, the rollup full nodes execute the transactions and send state commitments to a smart contract.
If an invalid commitment appears on the contract, honest full nodes post fraud proofs to warn the light clients about the invalid state.
To guarantee that they will see a fraud proof on-chain when the state is invalid, these clients wait for a \emph{dispute delay} period, typically one week, before accepting a rollup state commitment.
Thus, Ethereum acts like a lazy blockchain towards the pending rollup transactions and state commitments that are less than a week old.
With our light client protocol used on these pending transactions, clients can sync with the latest rollup state within seconds.

\noindent
\textbf{Contributions.}
Our contributions 
are:
  (1) We put forth the first \emph{light client} for \emph{lazy blockchains}, achieving
      exponential improvement over full nodes in terms of communication and computational complexity (Section~\ref{sec:protocol});
  (2) We show our system is \emph{complete}, \emph{sound}, and \emph{succinct} with reduction-based proofs (Section~\ref{sec:analysis});
  (3) We implement our scheme and measure its performance 
  (Section~\ref{sec:implementation}).

Experiments show that our light client construction can be efficiently implemented on commodity mobile hardware, and only requires slight, incremental changes to blockchain full nodes serving these clients. Specifically, to synchronize with the network, a light client connecting to 17 full nodes distributed across the world only downloads a dozen MBs of data, as opposed to hundreds of GBs if running as a full node. The entire process takes less than 25 seconds.

\subsection{Construction Overview}
\label{sec:construction-overview}

Consider a light client connected to two full nodes, Charlie and Eve.
Charlie is honest and Eve is adversarial.
The client begins by downloading the \emph{canonical} (confirmed) header chain, each header containing the Merkle root of the transactions (valid and invalid) in its block.
To find out its account balance, the client queries the full nodes.
If both of them return the same answer,
it rests assured that the balance reported is accurate, implying that the protocol terminates quickly in the optimistic case.
Otherwise, it must identify the truthful party.

To help convince the light client, Charlie augments his dirty ledger with some extra information: Together with every transaction, he includes a state commitment
after the particular transaction has been applied to the previous state.
If a transaction is invalid, he does not update the state. 
He organizes this \emph{augmented dirty ledger} into a binary Merkle tree, the \emph{dirty tree}.
All honest full nodes following this process will construct the same dirty tree and hold the same, correct dirty tree root (assuming they claim the same number of leaves that is a power of two).
If the client somehow learns the \emph{correct} dirty tree root,
then it can be convinced of its balance with a Merkle proof.
Thus, it suffices for the client to discover the correct root.
(In practice, the dirty tree can be organized on the granularity of blocks rather than transactions; \cf Section~\ref{sec:state-transition}.)

Charlie gives the correct dirty tree root to the light client, whereas
Eve gives an incorrect root.
Since the two roots are different,
the underlying augmented dirty ledgers must differ somewhere.
Charlie helps
the client identify the first leaf in Eve's dirty tree
that differs from his own via a \emph{bisection game} (\cf
\ifappendices Appendix~\ref{sec:bisection-game} \else\cite[Appendix B]{full-version} \fi
for a formal description):
With reference to his own dirty tree,
he guides the client through a path on Eve's dirty tree that starts at the
root and ends at the first leaf of disagreement.
He does this
by iteratively asking Eve to reveal an increasingly deeper node at a time.
Given a node revealed at a certain height, Charlie queries the left or
the right child as illustrated in Figure~\ref{fig:bisection_game}.
The left child is queried if it does not match the corresponding internal node of his own tree,
indicating a mismatch; otherwise he selects to query the right child,
since the left subtrees are identical.
When the process finishes, the
light client has arrived at the first point of disagreement between Charlie's
and Eve's augmented dirty ledgers.

Once the augmented dirty ledger entry of disagreement is identified, the client must verify that Eve's entry is fraudulent, as claimed by Charlie:
It either contains an incorrect transaction or an invalid state commitment.
If the transaction within Eve's entry is different from the one in the confirmed dirty ledger at the claimed position,
the client can detect this by asking for the transaction's Merkle inclusion proof with respect to the header chain 
client already holds 
(\cf \ifappendices Appendix~\ref{sec:consensus-oracle} \else\cite[Appendix E.2]{full-version} \fi
for a formalization of this inclusion check and \ifappendices Appendix~\ref{sec:coracle-constructions} \else\cite[Appendix F]{full-version} \fi for an overview of how it can be implemented on different consensus protocols).
On the other hand, if the transaction is correct, the client can locally evaluate the correct state commitment 
at that position by applying the transaction to the previous state commitment (which is valid as Charlie agrees with it).
For this purpose, the client need not download the whole previous state tree, but instead asks for a fraud proof from Charlie.
Fraud proofs were first introduced by Al-Bassam et al.~\cite{albassam2018fraud} 
to provide security for light clients against dishonest \emph{majorities} that can include invalid state commitments in Ethereum.
They consist of the state elements touched by the transaction and their Merkle proofs within the previous state commitment.
They allow the client to obtain the new, correct state commitment by updating these state elements and the relevant inner nodes of the state Merkle tree.
Therefore, any discrepancy in the state commitment can be caught by the client (\cf \ifappendices Appendix~\ref{sec:execution-oracle} \else\cite[Appendix E.3]{full-version} \fi for a formalization of this state check and how it can be implemented on UTXO based protocols).

\subsection{Related work}
\label{sec:related-work}

\emph{Eager} light clients were first introduced by Naka\-moto~\cite{bitcoin}.
\emph{Superlight} clients for eager proof-of-work blockchains were put forth as
NIPoPoWs \cite{nipopows,flyclient,compactsuperblocks,pow-sidechains} and Mina~\cite{coda} (formerly known as CODA).
Improved and superlight clients for eager proof-of-stake chains were described
by Gaži et al.~\cite{pos-sidechains} and Agrawal et al.~\cite{popos}\footnote{The paper~\cite{popos} was made public after an initial version of this work.} respectively.
Our techniques are orthogonal to theirs and can be composed as discussed in
\ifappendices Appendix~\ref{sec:super_light_clients}\else\cite[Appendix D]{full-version}\fi.
For an overview of different light client constructions, we refer the reader to Chatzigiannis et al.~\cite{light-clients-sok}.

Our interactive bisection game is based on the work of Canetti et al.~\cite{refereed-computation}
which was first applied in the blockchain setting by Arbitrum~\cite{arbitrum}.
Computation over large public logs was also explored by VerSum~\cite{versum}.
Contrary to Canetti and Arbitrum, where bisection games are used to
dispute \emph{computation} over static data, our
bisection games are administered over \emph{ledgers}, ever-growing and with different sizes.
This challenge
requires us to introduce novel techniques in these refereed games such as the use of Merkle Mountain Ranges~\cite{mmr,mmr-grin} and a \emph{Suffix Monologue} in our construction (Section~\ref{sec:full-protocol}).
Finally, on the multi-server case, we improve the quadratic communication complexity of Canetti~\cite{refereed-computation} to linear by our multiparty tournaments (Section~\ref{sec:tournament}).

As an alternative to our construction, recursive compositions~\cite{recursiveSNARK} of SNARKs \cite{sasson-recursive-snark,coda} or STARKs can be used to support non-interactive lazy light clients.
For example, Mina~\cite{coda} relies on recursive SNARKs with trusted setup to enable verification of \emph{all} past state transitions in constant time.
Plumo~\cite{plumo} proposes a SNARK-based blockchain client with trusted setup that can prove months of state history with a single transition proof.
Halo~\cite{halo} (later formalized by~\cite{bunz-recursive-proof}) introduced the first practical recursive SNARK without trusted setup.
Our work also does not require a trusted setup and our provers can update their state in an online fashion within milliseconds on commodity hardware, with minimal RAM requirements (for comparison, zkBridge~\cite{zkbridge} that uses SNARK proofs incurs a proving cost of $\$50$ million per year).
We also do not require changes in the consensus layer to support pairing-friendly and ZK-friendly elliptic curves.
Our construction uses simple primitives that are straightforward to implement today, and give insight to the structure of the underlying problem.
Lastly, although the ZK-based solutions do not require synchrony and the existential honesty assumption for the safety of the lazy light clients (albeit requiring them for liveness), these assumptions are already needed for the clients to identify the correct header chain (consensus security) upon bootstrapping on many blockchains such as Bitcoin, Ethereum and Cardano.
Therefore, our work does not introduce extra assumptions for the security of the lazy light client construction.

%% file: model.tex
\noindent
\textbf{Notation.}
For a natural number $n$,
we use $[n]$ to denote the set $\{1, \cdots, n\}$.
We use $\epsilon$ for the \emph{empty string}. Given two strings $a$ and $b$, we write $a \concat b$
for some unambiguous encoding of their \emph{concatenation}.
Given a sequence $X$, $X[i]$ represents the $i^\text{th}$ element (starting from $0$).
Negative indices address elements from the end, 
so $X[-1]$ is the last element. We use $X[i{:}j]$ to denote the subsequence of $X$
consisting of the elements indexed from $i$ (inclusive) to $j$ (exclusive). The notation $X[i{:}]$ means the
subsequence of $X$ from $i$ onwards, while $X[{:}j]$ means the subsequence of $X$ up to (but excluding) $j$.
We use $|X|$ to denote the size of a sequence. For a non-empty sequence $X$, we use $(x{:}xs) \gets X$ to mean that
the first element of $X$ is assigned to $x$, while the rest of the elements are assigned to the (potentially empty)
sequence $xs$. 
In our multi-party algorithms, we use $m \dashrightarrow A$ to indicate that message $m$ is
sent to party $A$ and $m \dashleftarrow A$ to indicate that message $m$ is received from party $A$.
We use $X \preceq Y$ ($X \prec Y$) to mean that $X$ is a (strict) prefix of $Y$.
If either $X \preceq Y$ or $Y \preceq X$, then $X$ and $Y$ are said to be consistent.
We use $X\mid Y$ to denote that $X$ is a \emph{subarray} of $Y$, \ie, all elements in $X$ appear in $Y$ consecutively.
We use $H$ to denote a generic, \emph{collision-resistant} cryptographic hash function~\cite{katz}. 

There are three types of nodes: \emph{consensus nodes}, \emph{full nodes}, and \emph{light clients}.

\smallskip
\noindent
\textbf{Consensus nodes} receive \emph{constant} size transactions from the network and run a consensus protocol to obtain chains consisting of blocks.
These chains are subsequently broadcast to all other nodes.
Upon receiving a \emph{confirmed} chain from the consensus nodes, each node reads its chain and produces a sequence of transactions (with total order) called the \emph{ledger}.

The consensus nodes are \emph{lazy}:
They treat transactions as meaningless strings, without validating them.
They include in their proposed blocks \emph{any} received transaction
with some spam-resilience mechanism (\eg, they typically maintain a minimal notion of state that enables transactions to pay fees for block space).

The ledgers held by different nodes satisfy two properties:
(1) \emph{Safety}
mandates that the ledgers of all honest nodes are consistent with each other;
(2) \emph{Liveness} mandates that, if an
honest node broadcasts a new transaction, it will eventually
appear in the ledger of all honest nodes within some finite delay.

\smallskip
\noindent
\textbf{Full nodes} 
do not execute the consensus protocol, and instead, rely on the consensus nodes to provide them with a confirmed chain and the associated ledger. 
Contrary to consensus nodes, full nodes
execute transactions to maintain a \emph{state} (\eg, a Merkle tree of account balances)
uniquely determined by the ledger.
An empty ledger corresponds to a constant \emph{genesis state}, $\genesisstate$.
To determine the state of a non-empty ledger, transactions from the ledger
are iteratively applied on top of the state, starting at the genesis state.
This 
is captured by a transition
function $\transition(\cdot, \cdot)$ taking a state and a transaction and producing
a new state. 
Given a dirty ledger $\ledger = \tx_1 \cdots \tx_n$, the state becomes
$\transition(\transition(\cdots \transition(\genesisstate, \tx_1), \cdots), \tx_n)$.
We use the shorthand notation $\transition^*$ to apply a sequence of transactions
$\overline{\tx} = \tx_1 \cdots \tx_n$
to a state, \ie,
$\transition^*(\genesisstate, \overline{\tx}) =
\transition(\transition(\cdots \transition(\genesisstate, \tx_1), \cdots),\allowbreak\tx_n)$.

Some transactions may not be applicable to a particular state, in which case they are said to be \emph{invalid} with respect to the state (\eg, double spends).
As we are dealing with \emph{lazy} systems, invalid
transactions may still be contained in the ledger, hence the ledger
is called \emph{dirty}.
We denote by $\LOGdirty{r}{v}$ the dirty ledger in the view of a full node $v$ at time $r$.
If safety is guaranteed, then we use $\ledgercup_r$ to denote
the longest among all the dirty ledgers kept by honest nodes at round $r$.
Similarly, we use $\ledgercap_r$ to denote the shortest among them.
We skip $r$ in this notation if it is clear from the context.
By convention, if a transaction $\tx$
cannot be applied to state $\st$,
we let $\transition(\st, \tx) = \st$. 
Each state is committed to by a succinct representation called the \emph{state commitment} (\eg, a Merkle root) and denoted by $\stc$.
State commitments have constant size.
We denote by $\left<\cdot\right>$
the commitment function that takes a state and produces its commitment, \ie, $\stc$ is the commitment to the state $\st$.

\smallskip
\noindent
\textbf{Light clients} wish to find out a particular state element (\eg, its own account balance) 
without downloading the whole ledger or executing the transactions.
As in the SPV model, the light client downloads and verifies the \emph{header chain} from the consensus nodes (\eg, the longest chain headers containing \emph{transaction roots}), but not the transactions themselves. 
Given a chain $\chain$ with $|\chain|$ blocks and the corresponding ledger of size $L = |\ledger|$, a full node downloads
data proportional to $\mathcal{O}(|\chain| + L)$, where $|\chain|$ comes from the header chain and $L$ comes from the transactions.
In contrast, a light client wants to learn its desired state element by downloading asymptotically less data. 
We call a light client \emph{succinct} if instead of $L$, it only needs to download $O(\text{poly}\log L)$ bits after obtaining the header chain.

\smallskip
\noindent
\textbf{The Prover--Verifier model.}
We are interested in a \emph{light client} $\verifier$ who is booting up the network for the first time.
It connects to full nodes who are fully synchronized with the rest of the network.
The client acts as a \emph{verifier}, while the full nodes act as \emph{provers}~\cite{nipopows}.
We assume at least one of the provers that $\verifier$ connects to is honest
(the standard \emph{non-eclipsing} assumption~\cite{backbone,varbackbone,eclipse,eclipse-ethereum}),
but the rest can be adversarially controlled.
The honest provers follow the specified protocol and the adversary can run any probabilistic polynomial-time algorithm.

\smallskip
\noindent
\textbf{Network.}
Time proceeds in discrete rounds.
The network is \emph{synchronous}, \ie, a message sent by one honest node
at the end of round $r$ is received by all honest nodes at the beginning of round $r + 1$. 
The adversary can inject arbitrary, but bounded number of messages to the network.
She can also reorder the
messages sent by honest nodes and deliver them in a different order to different honest nodes.
However, she cannot censor honest messages.
As popular lazy blockchain systems such as Celestia (LazyLedger)~\cite{lazyledger},
Prism~\cite{prism}, and Parallel Chains~\cite{parallel-chains}
were proven secure under the synchronous network model, our construction does not impose extra requirements for these systems.

%% file: protocol.tex
We next describe the protocol for ledgers of variable and dynamic lengths.
\subsection{Augmented Dirty Ledgers, Dirty Trees and MMRs}
\label{sec:constructing-dirty-ledgers}

The prover \emph{augments} each element of its dirty ledger $\ledger$ and produces an \emph{augmented dirty ledger} $\ledgeraug$, where every transaction in the original dirty ledger is replaced with
a pair.
The pair, denoted by $(\tx, \stc)$, contains the original transaction $\tx$ as well as a commitment $\stc$ to the state after this transaction is applied.
The first element of $\ledgeraug$ is the pair $(\epsilon, \genesisstatec)$, consisting of the empty string (as there is no genesis transaction) and the genesis state commitment $\genesisstatec$.
The state commitment of $\ledgeraug[i + 1]$ is computed by applying the transaction at $\ledger[i+1]$ to the state committed to by $\ledgeraug[i]$.
Concretely, if $\ledger = (\tx_1, \tx_2, \ldots)$, then
$\ledgeraug =  (
               (\epsilon, \genesisstatec),
               (\tx_1, \left<\transition(\genesisstate, \tx_1)\right>),
               (\tx_2, \left<\transition(\transition(\genesisstate, \tx_1), \tx_2))\right>),
               ..
              )$.
              
The dirty tree $\dtreesp$ corresponding to an augmented dirty ledger $\ledgeraug$ is defined as the Merkle tree that contains $\ledgeraug[i]$ as the $i$-th leaf.

To organize ledgers of various sizes, provers use Merkle Mountain Ranges (MMRs). 
Provers construct their MMRs on their augmented dirty ledgers.
To build an MMR, a prover divides his augmented dirty ledger $\ledgeraug$ into segments $s_1, s_2, \ldots, s_k$ with lengths $\overline{\ell} = (\ell_1, \ell_2, \ldots, \ell_k)$, where $\ell_1 = 2^{q_1} > \ell_2 = 2^{q_2} > \ldots > \ell_k = 2^{q_k}$ are unique decreasing powers of $2$.
Each of these $k$ segments is then organized into a dirty tree, and those trees $\dtreesp = (\dtreesp_1, \dtreesp_2, \ldots, \dtreesp_k)$ 
are collected into an array, that is the MMR $\dtreesp$.
The roots $\mroot = (\mroot_1, \mroot_2, \ldots, \mroot_k)$ of these dirty trees, where $\mroot_i = \left<\dtreesp_i\right>$,
are called the \emph{peaks}.
When there is a new transaction, the provers update their MMRs in amortized constant time, worst case update time per transaction being logarithmic in the size of the dirty ledger.

\subsection{Views in Disagreement}
\label{sec:views-in-disagreement}

Consider a light client $\verifier$ that connects to an honest prover $\prover$ and an adversarial prover $\prover^*$, but does not know who is who.
Let $\st$ be the current state in $\prover$'s view at round $r$.
Let $\ledger$ denote the dirty ledger, $\ledgeraug$ denote the augmented dirty ledger, and $\dtreesp$ denote the
MMR of $\prover$ at round $r$.
The last entry $\ledgeraug[-1]$ of the honest augmented dirty ledger
contains the commitment $\stc$ to the latest state $\st$.
The client $\verifier$ wishes to learn the value of a particular state element in $\st$.
For this purpose,
$\verifier$ only needs to learn a truthful state commitment $\stc$;
as from there, an inclusion proof into $\stc$ suffices to show inclusion of any state element value.
So the goal of $\prover$ is to convince $\verifier$ of the correct state commitment $\stc$.
If both provers respond to $\verifier$'s request with the same commitment $\stc$, then $\verifier$ knows
that the received state commitment is correct (because at least one prover is honest).
If they differ, it must discover the truth. 
If at any point in time, one of the provers \emph{timeouts}, \ie, fails
to respond in one round of receiving its query, the prover is considered adversarial
and ignored thereafter (as the network is synchronous, no honest prover would timeout). 
In practice, this equates to a short timeout in the network connection.

Suppose the two provers claim different state commitments, $\stc$ and $\stc^*$ respectively, where $\stc \neq \stc^*$.
To prove the correctness of its commitment, $\prover$ sends to $\verifier$ the peaks $\mroot = (\mroot_1, \ldots, \mroot_k)$ of its MMR $\dtreesp$,
the length $\ell=|\ledgeraug|$ of its augmented dirty ledger, 
and the Merkle proof $\pi$ from the last leaf $\ledgeraug[-1]$, which contains $\stc$, to the root $\mroot_k$.
The adversary $\prover^*$ sends to $\verifier$ the alleged peaks $\mroot^* = (\mroot^*_1, \ldots, \mroot^*_{k^*})$, an alleged length $\ell^*$,
and an alleged Merkle proof $\pi^*$.
Since $\stc \neq \stc^*$, if $\pi$ and $\pi^*$ both verify, then we have that
$\mroot_k \neq \mroot^*_{k^*}$.
In this case, $\verifier$ mediates a \emph{challenge game} between $\prover$ and $\prover^*$ to determine
which of the peaks $\mroot$ or $\mroot^*$ were constructed honestly:

\begin{definition}[\Wf Ledgers, Trees and MMRs]\label{def:safe}
An augmented dirty ledger $\ledgeraug$ is said to be \emph{\wf} at round $r$ with respect to transition $\transition$, genesis state $\st_0$, and commitment function $\left<\cdot\right>$
if: $\ledgeraug[0]=(\epsilon, \genesisstatec)$ and,
$\forall \ i \in [|\ledgeraug|-1]$, $\ledgeraug[i]=(\tx_i,\left<\st_i\right>)$ such that $(\tx_{i-1}, \tx_i) \mid \ledgercup_r$, and $\st_i = \transition^*(\genesisstate,\ledger[{:}i])$.

A dirty tree or MMR $\dtreesp$ is said to be \emph{\wf} if its leaves correspond to the entries of a \wf augmented dirty ledger.
\end{definition}

The augmented dirty ledger and MMR held by an honest prover are always \wf.
Hence, to determine whether $\mroot_k$ or $\mroot^*_{k^*}$ contain the correct state commitment, it suffices for the verifier to check
if $(\mroot_1, \ldots, \mroot_{k})$ or $(\mroot^*_1, \ldots, \mroot^*_{k^*})$ correspond to the peaks of a \emph{\wf} MMR.

\subsection{Challenge Game}
\label{sec:full-protocol}

\begin{figure}[t]
    \centering
    \ifonecolumn
    \includegraphics[width=0.8 \textwidth]{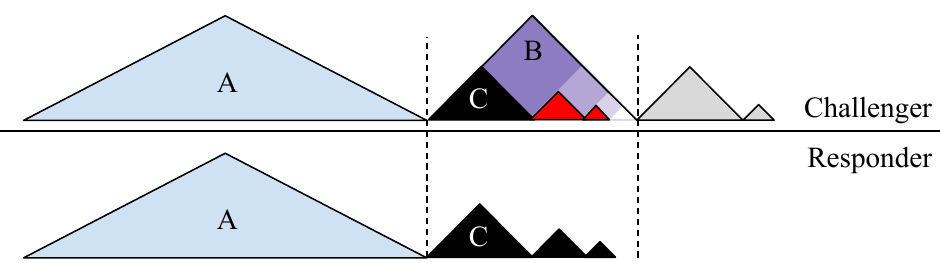}
    \fi
    \iftwocolumn
    \includegraphics[width=\linewidth]{figures/tree-vs-peaks.pdf}
    \fi
    \caption{The challenger's MMR (top) is compared to the responder's alleged MMR.
             The first two peaks (A in blue) are the same,
             so they are skipped by Alg.~\ref{alg.peaks.vs.peaks}.
             The second peak of the challenger is reached
             (B in purple) and compared against the responder's
             second peak (C). When found to be
             different, the challenger knows that the remaining
             responder peaks (in black, bottom) will lie within
             his own current tree (B, in purple); so
             Alg.~\ref{alg.peaks.vs.peaks} Line~\ref{alg.peaks.vs.peaks.subtree}
             calls Alg.~\ref{alg.challenger.tree.vs.peak}
             to compare the black peaks against the purple tree.}
    \label{fig.tree-vs-peaks}
\end{figure}

We now explore the challenge game that allows the verifier to compare competing claims by two provers.
During the game, the prover with the larger claimed ledger length $\ell$
acts as the challenger while the other acts as the responder.
The goal of the challenger is to identify the first point on the responder's alleged augmented dirty ledger that disagrees with his own ledger.
The challenge game consists of two phases:
During the first \emph{zooming phase}, the challenger reduces his search of the first point of disagreement
to a single tree within the responder's MMR. 
After this first phase is completed, the second phase consists of either the two parties playing a \emph{bisection game} (\cf Section~\ref{sec:construction-overview}, and for a more detailed description \ifappendices Appendix~\ref{sec:bisection-game}\else \cite[Appendix B]{full-version}\fi) or the challenger going into a \emph{suffix monologue}.

\noindent
\textbf{Zooming phase.}
To narrow his search down to a single tree, the challenger first calls Alg.~\ref{alg.peaks.vs.peaks} to identify the earliest peak
among the responder's peaks that disagrees with his own peaks.
Alg.~\ref{alg.peaks.vs.peaks} iterates over the responder's peaks (Alg.~\ref{alg.peaks.vs.peaks} Line~\ref{alg.peaks.vs.peaks.iterate}) until the challenger finds a peak $\mroot^*_i$ among those returned by the responder, that is different from the corresponding root $\left<\dtreesp_i\right>$ in his own peaks.
If the number of leaves under both peaks are the same, the challenger plays a bisection game on the Merkle trees whose roots are $\left<\dtreesp_i\right>$ and $\mroot^*_i$ (Alg.~\ref{alg.peaks.vs.peaks} Line~\ref{alg.peaks.vs.peaks.challenge}).
Otherwise, if the number of leaves under $\mroot^*_i$ is smaller than the number of leaves under $\left<\dtreesp_i\right>$,
then, all the alleged data within the responder's remaining peaks lies under the $i^\text{th}$ peak of the challenger
(see Figure~\ref{fig.tree-vs-peaks}).
The challenger has now reduced his search to his own $i^\text{th}$ tree
and can compare it against the responder's remaining peaks. This is done by calling
Alg.~\ref{alg.challenger.tree.vs.peak} on the remaining peaks of the responder (Alg.~\ref{alg.peaks.vs.peaks} Line~\ref{alg.peaks.vs.peaks.subtree}).

\import{./}{algorithms/alg.peaks.vs.peaks.tex}

\import{./}{algorithms/alg.challenger.tree.vs.peak.tex}

\import{./}{algorithms/alg.respond.tex}

Alg.~\ref{alg.challenger.tree.vs.peak} narrows the search for the first point of disagreement to one of responder's peaks, so that $\verifier$ can compare the two trees using a bisection game.
Consider the responder's remaining peaks overlayed onto the challenger's tree $\dtreesp_i$ (dashed lines in Figure~\ref{fig.tree-vs-peaks}).
They correspond to certain inner nodes within $\dtreesp_i$ (black, red, and red subtrees at the top of Figure~\ref{fig.tree-vs-peaks}).
Alg.~\ref{alg.challenger.tree.vs.peak} locates the first such inner node that disagrees with the responder's corresponding peak
(the left-most red subtree in Figure~\ref{fig.tree-vs-peaks}).
Finally, at this point, the challenger plays the bisection game on the sub-trees under this inner node and the responder's currently inspected root (Alg.~\ref{alg.challenger.tree.vs.peak} Line~\ref{alg.challenger.tree.vs.peak.5}). 
After the bisection game, either the challenger or the responder is declared the winner and the other one is declared the loser.

\noindent
\textbf{Suffix monologue.}
When the MMRs are well-formed,
there is no first point of disagreement between the two alleged augmented dirty ledgers, and the ledgers form a prefix of one another.
In that case, the provers will not enter into a bisection game,
and it is the challenger's turn to present his augmented dirty ledger entries extending the responder's ledger with size $\ell^*$. 
Specifically,
the challenger presents the \emph{suffix} $\ledgeraug[\ell^*{:}\min(\ell,\ell^* + \psi)]$ and the verifier checks the transitions within this suffix.
Concretely, for every consecutive $(\tx_j,\left<\st_j\right>)$ and
$(\tx_{j+1},\left<\st_{j+1}\right>)$, for $\ell^* \leq j < \min(\ell,\ell^* + \psi)$,
the verifier checks the inclusion of $\tx_j$ and $\tx_{j+1}$ in the header chain as before, and verifies that the state $\left<\st_{j+1}\right>$ has been computed correctly using $\transition$.
The verifier also checks the transition from $\ledgeraug^*[\ell^*-1]=\ledgeraug^*[-1]$, \ie, the responder's last entry, to $\ledgeraug[\ell^*]$, \ie, the first entry in the challenger's suffix, since the challenger, by starting the suffix monologue, claims that his augmented dirty ledger is a suffix of the responder's.

The parameter $\psi$ is a constant selected in accordance with the chain growth and liveness parameters of the blockchain (\cf \ifappendices Appendix~\ref{sec:generalize}\else\cite[Appendix E]{full-version}\fi).
The bound $\psi$ on the number of transitions to check prevents the suffix monologue from violating succinctness.
By the Common Prefix property~\cite{backbone}, discrepancy in the lengths of two honest provers' ledgers is bounded when there is an upper bound on the chain growth rate, which is the case for our protocols of interest (\cf the ledger Lipschitz property in \ifappendices Appendix~\ref{sec:consensus-protocol}\else\cite[Appendix E.1]{full-version}\fi).
Hence, if the challenger presents $\psi$ or more extra entries with consecutive transactions and correct state transitions,
then the responder is declared a loser, as he presented too short a ledger to possibly be honest.
In other words, if an adversarial responder presents a much shorter ledger, then the honest challenger sends $\psi$ entries, proving to the verifier that the adversary's ledger is too short.
On the contrary, an adversarial challenger cannot present a \wf ledger much longer than an honest responder's ledger, without breaking the underlying consensus protocol.
If the challenger fails to present a \wf suffix, then the responder is declared the winner,
while the challenger is declared the loser.
Otherwise, if the suffix presented is \wf and has length less than $\psi$, then \emph{both} provers are declared winners of
the challenge game.
Unlike the bisection game, at the end of the suffix monologue,
\emph{both} the challenger and the responder can win.

\subsection{Multiparty Tournaments}
\label{sec:tournament}
\input{tournament}

%% file: algorithms/alg.peaks.vs.peaks.tex
\begin{algorithm}[t]
  \caption{The algorithm run by the challenger to identify the first peak in the responder's MMR that is different from that of the challenger. The variables $\challenger{\dtreesp}$ and $\challenger{\overline{\ell}}$ denote the challenger's sequence of Merkle trees and a sequence of their respective sizes, whereas $\responder{\mroot^*}$ and $\responder{\overline{\ell}^*}$ denote the responder's sequence of peaks and the corresponding number of leaves respectively.  The algorithm $\textsc{BisectionGame}$ initiates a bisection game between the challenger's tree and the responder's alleged tree with the same size.}
  \label{alg.peaks.vs.peaks}
  \begin{algorithmic}[1]\small
    \Function{\sc \challenger{PeaksVsPeaks}}{$\challenger{\dtreesp}, \challenger{\overline{\ell}}, \responder{\mroot^*}, \responder{\overline{\ell}^*}$}
        \For{$i = 0$ to $|\responder{\mroot^*}| - 1$}\label{alg.peaks.vs.peaks.iterate}
            \If{$\challenger{\overline{\ell}}[i] \neq \responder{\overline{\ell}^*}[i]$}
                \State \Return $\challenger{\textsc{TreeVsPeaks}}(\challenger{\dtreesp}[i], \responder{\mroot^*}[i{:}], \responder{\overline{\ell}^*}[i{:}])$ \label{alg.peaks.vs.peaks.subtree}
            \EndIf
            \If{$\challenger{\dtreesp}[i]\textrm{.root} \neq \responder{\mroot^*}[i]$}
                \State \Return $\challenger{\textsc{BisectionGame}}(\challenger{\dtreesp}[i], \challenger{\overline{\ell}}[i])$ \label{alg.peaks.vs.peaks.challenge}
            \EndIf
        \EndFor
    \EndFunction
  \end{algorithmic}
\end{algorithm}

%% file: algorithms/alg.challenger.tree.vs.peak.tex
\begin{algorithm}[t]
  \caption{The algorithm run by the challenger to identify the first subtree, under one of the challenger's larger Merkle trees, that is different from the responder's peak. The variable $\challenger{\dtreesp}$ denotes the challenger's larger Merkle tree whereas $\responder{\mroot^*}$ and $\responder{\overline{\ell}^*}$ denote the responder's sequence of peaks (with some of the first elements chopped off during the recursion) and the corresponding number of leaves respectively.}
  \label{alg.challenger.tree.vs.peak}
  \begin{algorithmic}[1]\small
    \Function{\sc \challenger{TreeVsPeaks}}{\challenger{\dtreesp}, \responder{$\mroot^*$}, \responder{$\overline{\ell}^*$}}
         \Assert{\textrm{\challenger{\dtreesp}.size} > \sum_{\ell^* \in \overline{\ell}^*} \responder{\ell^*}}

         \If{$|\responder{\mroot^*}| = 0$}
             \State\Return\Comment{Done: the MMRs are \wf.} \label{alg.challenger.tree.vs.peak.no.bisection}
         \EndIf
         \Let{(\responder{\textrm{peak}}{:}\responder{\textrm{peaks}})}{\responder{\mroot^*}}
         \Let{(\responder{\textrm{reSize}}{:}\responder{\textrm{reSizes}})}{\responder{\overline{\ell}^*}}
         \If{$\lfloor\frac{\challenger{\dtreesp}\textrm{.size}}{2}\rfloor > \responder{\textrm{reSize}}$} \label{alg.challenger.tree.vs.peak.1}
             \State $\challenger{\textsc{TreeVsPeaks}}(\challenger{\dtreesp}\textrm{.left}, \responder{\mroot^*}, \responder{\overline{\ell}^*})$ \label{alg.challenger.tree.vs.peak.2}
         \ElsIf{$\textrm{\challenger{tree}.left.root} = \textrm{\responder{peak}}$} \label{alg.challenger.tree.vs.peak.3}
             \State $\textsc{\challenger{TreeVsPeaks}}(\challenger{\dtreesp}\textrm{.right}, \responder{\textrm{peaks}}, \responder{\textrm{reSizes}})$ \label{alg.challenger.tree.vs.peak.4}
         \Else
             \State \textsc{\challenger{BisectionGame}}$(\challenger{\dtreesp}, \responder{\textrm{reSize}})$ \label{alg.challenger.tree.vs.peak.5}
         \EndIf
    \EndFunction
  \end{algorithmic}
\end{algorithm}

%% file: algorithms/alg.respond.tex
\begin{algorithm}[t]
  \caption{The algorithm ran by the responder to reply to the challenger's queries while the challenger tries to identify the first point of disagreement against the responder's MMR. The variable $\responder{\ledgeraug}$ denotes the responder's augmented dirty ledger. The algorithm $\textsc{MakeMMR}$ returns the MMR based on the given augmented dirty ledger.}
  \label{alg.respond}
  \begin{algorithmic}[1]\small
    \Function{\sc \responder{Respond}}{\responder{\ledgeraug}}
        \Let{\rm \responder{\textrm{trees}}}{\textsc{MakeMMR}(\responder{\ledgeraug})}
        \Let{\responder{\textrm{peaks}}}{\{\textrm{\responder{tree}.root}: \textrm{\responder{tree}} \in \responder{\textrm{trees}}\}}
        \Send{\textrm{\responder{peaks}}}{\challenger{\textsc{Challenger}}}
        \Receive{\challenger{\textrm{pNum}}}{\challenger{\textsc{Challenger}}}
        \Let{\responder{\textrm{tree}}}{\responder{\textrm{trees}}[\challenger{\textrm{pNum}}]}\Comment{Enter a particular Merkle Tree}
        \Let{\challenger{\mathrm{loc}}}{\bot}
        \While{$\responder{\textrm{tree}}\textrm{.size} > 1$}
            \Send{(\textrm{\responder{tree}.left}, \textrm{\responder{tree}.right})}{\challenger{\textsc{Challenger}}}
            \Receive{\challenger{\textrm{dir}}}{\challenger{\textsc{Challenger}}}
            \If{$\challenger{\textrm{dir}} = 0$}
                \Let{\responder{\textrm{tree}}}{\responder{\textrm{tree}}\textrm{.left}}
            \Else
                \Let{\responder{\textrm{tree}}}{\responder{\textrm{tree}}\textrm{.right}}
            \EndIf
            \Let{\challenger{\textrm{loc}}}{\challenger{\textrm{loc}} \concat \challenger{\textrm{dir}}}
        \EndWhile
        \Send{\responder{\textrm{dirtyLedger}}[\challenger{\textrm{loc}}]}{\challenger{\textsc{Challenger}}}
    \EndFunction
    \end{algorithmic}
\end{algorithm}

%% file: tournament.tex
Upon joining the network, the verifier $\verifier$ contacts a subset of all available provers\footnote{
For instance, according to \url{https://github.com/bitcoin/bitcoin/blob/master/doc/reduce-traffic.md}, Bitcoin makes $8$ outbound full-relay connections.}
for queries.
If all of the responses are the same, $\verifier$ accepts the response as the correct answer.
If it receives different responses, $\verifier$ arbitrates a \emph{tournament}
among the provers that responded.
The tournament's purpose is to select a prover whose latest claimed state
is 
as up-to-date as the state obtained by applying the transition function iteratively on
the transactions in $\ledger^\cap$.

Suppose $\verifier$ hears back from $n$ provers.
Before the tournament, $\verifier$ orders the $n$ provers into a sequence $\prover_1, \prover_2, \ldots, \prover_n$ in an increasing order of their (claimed) augmented dirty ledger sizes.
This sequence dictates the order in which the provers play the bisection games.
Then, $\verifier$ starts the tournament that consists of $n$ steps (\cf \ifappendices Appendix~\ref{sec:bisection-game} \else\cite[Appendix B]{full-version} \fi for the algorithm run by $\verifier$).
Before the first step, it initializes the set $\mathcal{S} = \emptyset$.
At the end of each step $t$, $\mathcal{S}$ contains the provers that have engaged in at least one challenge game, and have not lost any by step $t$.

The tournament starts with a challenge game between $\prover_1$ and $\prover_2$, during which $\prover_1$ with the larger alleged augmented dirty ledger challenges $\prover_2$.
The winners are added to $\mathcal{S}$ and the tournament moves to the second step.
At each step of the tournament, the set $\mathcal{S}$ is updated to contain the winners so far.
Players may be removed from the set $\mathcal{S}$ if they lose, and new winners can be added
to $\mathcal{S}$ as they win. Each player $\prover$ is considered in order.
Let $\largest$ denote the prover in $\mathcal{S}$ that claims to have the largest augmented dirty ledger at a given step $i$.
The prover among $\{\prover,\largest\}$ with the larger alleged augmented dirty ledger challenges the other prover. 
Depending on the outcome of the challenge game, there are three cases:

\noindent
1. If both provers win, $\prover$ is added to $\mathcal{S}$ and the tournament moves to step $i+1$. 

\noindent
2. If $\prover$ loses, the tournament directly moves to step $i+1$ and $\mathcal{S}$ stays the same. 

\noindent
3. If $\prover$ wins and $\largest$ loses, $\largest$ is removed from $\mathcal{S}$.
Then, $\prover$ challenges the new $\largest$, the prover with the largest alleged size among those remaining in the set $\mathcal{S}$ (or vice versa).
This case is repeated until either one of cases (1) or (2) happens, or there are no provers left in $\mathcal{S}$. 
If the latter happens, $\prover$ is added to $\mathcal{S}$ and the tournament moves to step $i+1$. 

The procedure above is repeated until the end of step $n-1$, after which, $\largest$ wins the tournament.
Then, the verifier accepts the state commitment of $\largest$ among those remaining in $\mathcal{S}$ as the correct state.

The tournament consists of $\mathcal{O}(n)$ bisection games and has an $\mathcal{O}(n)$ running time\footnote{In contrast, the playoff in \cite[Appendix G]{refereed-computation} consists of $\mathcal{O}(n^2)$ games and has an $\mathcal{O}(n)$ running time due to games happening in parallel.}. 
The reason is that, after each game,
one party is eliminated from the winners, either by not being added to $\mathcal{S}$, or by being
removed from $\mathcal{S}$, and every party can be eliminated at most once.
Parallelizing the tournaments can also help make the runtime sublinear in the number of provers.

%% file: implementation.tex
\subsection{Running Time of Bisection Games}
\label{sec:implementation-results}
\noindent\textbf{Implementation and experimental setup.} We report on a prototype implementation of the prover and the verifier in 1000 lines of Go\footnote{
\ifnonanon
The code is open source under MIT license, and is available at \url{https://github.com/yangl1996/super-light-client}.
\fi
\ifanon
Anonymized source code repository is available at\\ https://anonymous.4open.science/r/light-client-lazy-blockchain/.
\fi
}, and set up 17 provers on AWS \texttt{r5.xlarge} instances distributed across 17 data centers around the globe. 
All provers have ledgers of the same size, but only one of the provers has the correct augmented dirty ledger. Ledgers of the remaining 16 provers differ from the correct one at a randomly selected position. 
To simplify the prototype, we do not implement a state transition function $\delta$, \eg, the Ethereum Virtual Machine. (We explore the cost of proving state transitions in the next subsection.) 
Instead, all transactions and state commitments (\cf \ifappendices Algorithm~\ref{alg.verifier}\else \cite[Algorithm 4]{full-version}\fi) are random byte strings. 
We hard code the prover and the verifier such that a state commitment is valid to the verifier only if the prover has the correct commitment in its augmented dirty ledger. 
We run the verifier under a residential internet connection with $300$ Mbps downlink and $10$ Mbps uplink bandwidth. 
The verifier connects to all of the 17 provers, and arbitrates the tournament (\cf \ifappendices Algorithm~\ref{alg.tournament}\else \cite[Algorithm 6]{full-version}\fi) among them.

\noindent\textbf{Verification latency.} We first explore the duration of the tournament. So far, we have only discussed \emph{binary} Merkle trees for use in our bisection game, but we can consider $m$-ary Merkle trees more generally. Increasing the degree $m$ of a dirty tree flattens the tree, resulting in fewer rounds of interactivity in the bisection game. On the other hand, opening an inner node now requires sending $m$ children, resulting in higher bandwidth usage. \ifappendices Appendix~\ref{sec:latency-bandwidth-tradeoff} \else\cite[Appendix C]{full-version} \fi models the latency-bandwidth trade-off realized by tuning $m$. Here, we experimentally explore the trade-off. 
For this experiment, we fix the ledger size to $10$ million transactions and vary $m$. 
For each configuration, we run $10$ tournaments and measure the average and the standard deviation of the duration (Figure~\ref{fig:implementation}a). 
When $m=300$, the duration reaches the lowest, at $18.37$s. 
Most blockchains confirm transactions with a latency of tens of seconds. 
In comparison, the tournament adds little extra time on top of the end-to-end latency that a light client perceives for new transactions. 
Under this configuration, a tournament consists of $16$ games, and each game involves $4$ rounds of interaction between the verifier and the provers. 
On average, each round of interaction lasts for $0.287$s. 
In comparison, the average network round-trip time (RTT) from the verifier to the provers is $0.132$s.

As we vary $m$, tournaments take longer to finish. 
Specifically, a smaller $m$ makes each Merkle tree opening smaller, but increases the number of openings per game, making network propagation latency the main bottleneck of the game. 
In contrast, a larger $m$ makes bandwidth the main bottleneck. 
As we increase $m$ to $10000$, each opening of the Merkle tree becomes large enough such that message transmission is affected by the fluctuation of the internet bandwidth, causing a higher variation in the tournament duration.

\noindent\textbf{Scalability.} We now evaluate how our scheme scales as the ledger size grows. 
We fix the tree degree $m$ to $300$, vary the size of the ledger from $1000$ to $100$ million transactions and report the mean and the standard deviation over $10$ tournaments for each datapoint (Figure~\ref{fig:implementation}b). 
As we increase the ledger size by $10^5\times$, the tournament duration grows from $13$s to $26$s, an increase of only $2\times$. This is because the number of interactions in a game is equal to the depth of the Merkle tree, which grows logarithmically with the ledger size.

\begin{figure*}[t]
\minipage[t]{0.33\textwidth}
  \centering
    \centering
    \includegraphics[width=\linewidth]{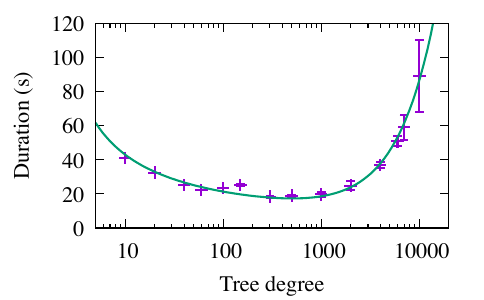}
    \label{fig:duration-exp}
\endminipage\hfill
\minipage[t]{0.33\textwidth}
   \centering
    \includegraphics[width=\linewidth]{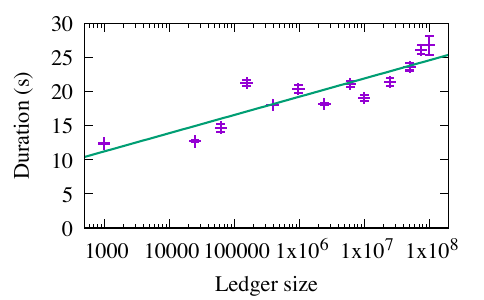}
    \label{fig:scale-exp}
\endminipage\hfill
\minipage[t]{0.33\textwidth}%
  \centering
    \includegraphics[width=\linewidth]{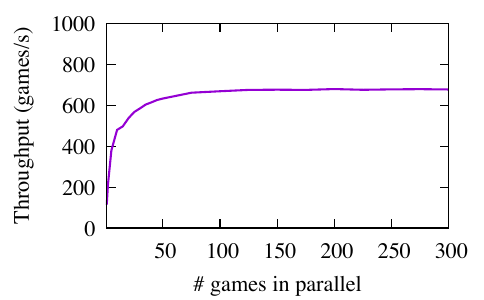}
    \label{fig:throughput-exp}
\endminipage
\caption{
  \label{fig:implementation}
  In (a) and (b), we measure the time to complete a tournament of 17 geo-distributed provers. Error bars show the standard deviation. Solid lines show the trend.
(a) Time when varying the tree degree $m$.
(b) Time when varying the ledger size $L$.
(c) Throughput of games with two provers and one verifier co-located in a data center. The verifier initiates games with variable parallelism to saturate the provers.}
\end{figure*}

\noindent\textbf{Prover throughput.} Finally, we evaluate the throughput of a prover participating in many games (Figure~\ref{fig:implementation}c). 
To minimize the network influence, we run two provers in the same datacenter. 
Each prover has a ledger of $10$M transactions, which differ at a random location. 
We start a verifier in the same datacenter, which initiates a variable number of bisection games between the two provers in parallel. 
We gradually ramp up the parallelism to generate enough load and saturate the provers. 
During the process, the achieved throughput first increases due to the increased load, and then stays flat because the provers have saturated their computational resources. 
Specifically, a prover running our prototype can support a throughput of $680$ games/second using its $4$ virtual CPU cores. 
We expect the throughput to scale with the available CPU cores and disk IO.

\subsection{Proving State Transitions}
\label{sec:state-transition}

We next discuss the cost of proving and verifying state transitions, which happens when a bisection game ends with a point of disagreement. For concreteness, we use the Ethereum Virtual Machine (EVM) as an example, but the discussion applies to other state machines.

\noindent\textbf{Ledger granularity.} So far, we have assumed that the bisection game runs with a granularity of transactions.
While the proof size is small in this natural configuration (less than 20 state elements on average for recent Ethereum transactions), an honest prover needs to maintain snapshots of ledger states as of every historical transaction to generate such proofs for arbitrary points in the dirty ledger. 
Maintaining these snapshots can be costly, since even blockchain nodes in ``archival'' mode---ones that store the most historical data---do not keep such fine-grained information.
We propose that real-world deployments use a granularity of \emph{blocks}, \ie, treating an entire block as a single entry in the dirty ledger. To generate state transition proofs, provers only need access to state snapshots as of each block, which
are readily available from archival nodes. 
A direct benefit is that provers can be implemented using public RPC APIs provided by EVMs, namely the \texttt{debug\_traceBlock} RPC which lists all state elements read/written by a block. 
This allows provers to make use of existing archival nodes and eliminates the need to maintain separate state snapshots.

An apparent downside of this coarse-grained approach is that state transition proofs are larger, consisting of state elements touched by an entire block plus the relevant Merkle proofs. 
However, our experiments show that such proofs are less than 10 MB for recent Ethereum blocks, which can be downloaded within 0.3 seconds with a 300 Mbps internet connection used in previous experiments, adding little to the seconds-long duration of the bisection game.

\noindent\textbf{Verification costs.} Upon downloading a state transition proof (consisting of the state elements touched by the transactions within the block at the first point of disagreement and their Merkle proofs), the verifier needs to check the proof by executing the transactions locally. We implemented a verifier by forking \texttt{foundry}\footnote{https://github.com/foundry-rs/foundry}, an EVM implementation in Rust, and used it to benchmark verification costs on commodity mobile hardware.

Experimental results show that verifying state transition of recent Ethereum blocks takes less than 0.8 seconds per block on average on a M1 MacBook Pro and consumes 2.5 Joules of energy\footnote{Measured using \texttt{powermetrics} built into macOS.}. The same verification takes less than 1.5 seconds on an underpowered tablet with a 2-core Intel m5 low-power CPU manufactured in 2015.
In comparison, a \emph{full node} syncing with the latest EVM state from genesis has to execute all historical transactions, which takes at least \emph{a full day} on a workstation with 32 GB of RAM and a 4-core Intel Xeon CPU, and uses 540 GB of SSD. Our construction saves significant time, computation, and storage because the light client only needs to locally execute the one block at the first point of disagreement.

%% file: analysis-simple.tex
We state our security theorems informally in this section.
For the rigorous theorem statements and proofs, see \ifappendices Appendix~\ref{sec:proofs}\else\cite[Appendix G]{full-version}\fi.
We begin by defining \emph{State Security}, which captures the verifier's goal of obtaining a state
consistent with the rest of the network: There is no disagreement with the other honest nodes
(safety), and the state downloaded is recent (liveness).

\begin{definition}[State Security]
\label{def:state-security}
  An interactive Prover--Verifier protocol $(P, V)$ is \emph{state secure} with safety parameter $\nu$, if
  there exists a ledger $\ledger$ such that the state commitment $\stc$ obtained by the verifier
  at the end of the protocol execution at round $r$ satisfies
  $\left<\transition^*(\st_0, \ledger)\right> = \stc$, and for all rounds $r' \geq r + \nu$: $\ledger$ is a \emph{prefix} of $\ledger^\cup_{r'}$ (safety) and $\ledger^\cap_r$ is a \emph{prefix} of $\ledger$ (liveness).

\end{definition}

The theorems for succinctness and security of the protocol are given below.
Security consists of two components: completeness and soundness.

\begin{lemma}[Succinctness (Informal)]
\label{thm:succinctness-informal}
The challenge game invoked at round $r$ with sizes $\ell_1$ and $\ell_2>\ell_1$ ends in $\mathcal{O}(\log(\ell_1))$ rounds of communication and
has, considered in isolation, a total communication complexity of $\mathcal{O}(\log r)$.
\end{lemma}

\begin{theorem}[Completeness (Informal)]
\label{thm:completeness-informal}
The honest responder wins the challenge game against any PPT adversarial challenger.
\end{theorem}

\begin{theorem}[Soundness (Informal)]
\label{thm:soundness-informal}
Let $H$ be a collision resistant hash function.
For all PPT adversarial responders $\mathcal{A}$, an honest challenger wins the challenge game against $\mathcal{A}$ with overwhelming probability in $\lambda$.
\end{theorem}

\begin{theorem}[Tournament Runtime (Informal)]
\label{thm:tournament-runtime-informal}
Consider a tournament started at round $r$ with $n$ provers.
Given at least one honest prover, for any PPT adversary $\mathcal{A}$, the tournament ends in $\mathcal{O}(n\log r)$
rounds of communication and has, considered in isolation, a total communication complexity of
$\mathcal{O}(n\log r)$, with overwhelming probability in $\lambda$.
\end{theorem}

The theorem below is a direct consequence of the above theorems.

\begin{theorem}[Security (Informal)]
\label{thm:security-informal}
Consider a tournament started at round $r$ with $n$ provers.
Given at least one honest prover, for any PPT adversary $\mathcal{A}$, the state commitment obtained by the prover at the end of the tournament satisfies \emph{State Security} with overwhelming probability in $\lambda$.
\end{theorem}

%% file: appendix_attack.tex
We examine a hypothetical attack on the succinctness of the communication complexity of the Bitcoin's SPV~\cite{bitcoin} on a lazy blockchain protocol with UTXO-based execution model.
Suppose at round $r$, the confirmed sequence of blocks contain the transactions $\tx_i$, $i \in [n]$ in the reverse order: for $i \in \{2,\ldots,n-1\}$, $\tx_i$ appears in the prefix of $\tx_{i-1}$. 
Each transaction $\tx_i$, $i \in [n]$, spends two UTXO's, $\UTXOP_i$ and $\UTXOL_i$, such that for $i=2,..,n$, $\UTXOL_i = \UTXOP_{i-1}$, and $\UTXOL_1$ $\UTXOP_i$ $i \in \{2,\ldots,n-1\}$, and $\UTXOP_n$ are all distinct UTXOs.
Thus, if $\tx_{i-1}$ appears in the prefix of $\tx_i$, it \emph{invalidates} $\tx_i$ as $\tx_i$ will be double-spending $\UTXOL_i = \UTXOP_{i-1}$ already spent by $\tx_{i-1}$.
However, $\tx_i$, $i \in \{3,\ldots,n\}$, does not invalidate $\tx_j$ for any $j<i-1$.
We assume that no transaction outside the set $\tx_{i}$, $i \in [n]$, invalidates $\tx_i$ for any $i \in [n-1]$.

If $n$ is even, $\tx_1$ would be invalid, since each transaction $\tx_i$ with an odd index $i \in \{1,3,\ldots,2n-1\}$, will be invalidated by the transaction $\tx_{i-1}$.
However, if $n$ is odd, each transaction $\tx_i$, with an even index $i \in \{2,4,\ldots,2n-1\}$, will be invalidated by the transaction $\tx_{i-1}$, which has an odd index.
Hence, $\tx_2$ would be invalid, implying that $\tx_1$ would be valid as no transaction other than $\tx_2$ can invalidate $\tx_1$ by our assumption.

Consider a light client whose goal is to learn whether $\tx_1$ is valid with respect to the transactions in its prefix.
Suppose $n$ is even, \ie, $\tx_1$ is invalid.
Towards its goal, the client asks full nodes it is connected to, whether $\tx_1$ is valid with respect to the transactions in its prefix.
Then, to convince the light client that $\tx_1$ is invalid, an honest full node shows $\tx_2$, which invalidates $\tx_1$, along with its inclusion proof.
However, in this case, an adversarial full node can show $\tx_3$ to the client, which in turn invalidates $\tx_2$, and gives the impression that $\tx_1$ is valid.
Through an inductive reasoning, we observe that the adversarial nodes can force the honest full node to show all of the transactions $\tx_i$ with even indices $i \in \{2,\ldots, n\}$, to the light client, such that no adversarial full node can verifiably claim $\tx_1$'s validity anymore.
However, since $n$ can be arbitrarily large, \eg, a constant fraction of the confirmed ledger length, light client in this case would have to download and process linear number of transactions in the ledger length.

A related attack on rollups that use a lazy blockchain as its parent chain is described by~\cite{woods-attack}. 

%% file: appendix_bisection_game.tex
\begin{figure*}
\fbox{
\begin{minipage}[t]{\linewidth}

The game is terminated by the verifier as soon as the victory of one of the provers becomes certain.

\noindent
\textbf{Challenger wins.}
The challenger wins the bisection game in the verifier's view whenever one of the following conditions fails:
\begin{enumerate}
    \item The responder must not timeout, \ie, must reply to a query within one round of receiving it from the verifier.

    \item The response of the responder must be syntactically valid according to the expectations of the verifier, \eg,
          if the challenger has asked for the two children of a Merkle tree inner node, these must be two hashes.

    \item For the two nodes $h_l$ and $h_r$ returned by the responder as the children of a node $h$ on its dirty tree, $h = H(h_l \concat h_r)$.

    \item If $j \geq 1$, the Merkle proof for the $(j-1)^\text{st}$ leaf of the responder's dirty tree is valid.

    \item \label{condition:challenger-consensus}
          If $j \geq 1$, $(\tx_{j-1}, \tx_j) \mid \ledger^{\cup}$.

    \item \label{condition:challenger-execution}
          If $j \geq 1$, for the claimed state commitments $\stc_{j-1}, \stc_j$, there is an underlying state $\st_{j-1}$
          such that $\left<\st_{j-1}\right> = \stc_{j-1}$ and $\stc_j = \left<\transition(\st_{j-1}, \tx_j)\right>$.

    \item If $j = 0$, the claimed state commitment $\stc_0$ matches the genesis state commitment $\left<\st_0\right>$ known to the verifier
          and $\tx_0 = \epsilon$.
\end{enumerate}

To check condition~\ref{condition:challenger-consensus}, the verifier consults the already downloaded header chain (\cf Appendix~\ref{sec:consensus-oracle}).
To check condition~\ref{condition:challenger-execution}, the verifier requests a \emph{proof} from the responder
that illustrates the correct state transition from $\stc_{j-1}$ to $\stc_j$, \eg, the balances that were
updated by $\tx_j$ (\cf Appendix~\ref{sec:execution-oracle}).

\noindent
\textbf{Responder wins.}
The responder wins and the challenger loses in the verifier's view if one of the following conditions fails:

\begin{enumerate}
    \item The challenger must send valid queries. A valid query is a root number on the first round (if the responder holds multiple Merkle trees), or a single bit in any next round.
    \item The challenger must not timeout, \ie, must send a query within a round of being asked by the verifier.
\end{enumerate}

If both parties respond according to these rules, then the responder wins.
\end{minipage}%
}%
\caption{The alg. ran by the verifier to determine the winner of the bisection game.}
\label{alg.verifier}
\end{figure*}

The bisection game is played between two provers, a \emph{challenger} and a \emph{responder}.
The challenger sends queries to the responder through the verifier. 
The responder replies through the same channel (\cf Figure~\ref{fig:bisection_game}).
The challenger sends his first query to the verifier.
The verifier forwards this query to the responder. 
The responder sends his response back to the verifier. 
Lastly, the verifier forwards this response to the challenger.
Subsequently, the challenger follows up with more queries. 
As the verifier forwards queries and responses, it checks they are well-formed and the responses correspond to the queries.
All communication between the two provers passes through the verifier.
The challenger's goal is to convince the verifier that the responder's dirty tree root does \emph{not} correspond to the root of a \wf tree.
The responder's goal is to defend his claim that his dirty tree root \emph{is} the root of a \wf tree.
We design the game so that an honest challenger always wins against an adversarial responder,
and an honest responder always wins against an adversarial challenger.

\import{./}{algorithms/alg.tree.vs.tree.tex}

The game proceeds as a \emph{binary search}~\cite{refereed-computation,practical-delegation,arbitrum}.
For simplicity, let us for now assume that $\ell = \ell^*$ and they are a power of two.
If $\mroot \neq \mroot^*$, then there must be a \emph{first point of disagreement}
between the two underlying augmented dirty ledgers $\ledgeraug$ and $\ledgeraug^*$ alleged by the two provers.
During the game, an honest challenger tries to identify the first point of disagreement between his augmented dirty ledger $\ledgeraug$
and the one the adversarial responder claims to hold.
Let $j$ be the index of that first point of disagreement pinpointed by the honest challenger.
Then, the challenger asks the responder to reveal the $(j-1)^\text{st}$ and $j^\text{th}$ entries of his augmented dirty ledger.
Upon observing that the revealed entries violate the \wfc conditions of Definition~\ref{def:safe},
the verifier concludes that the responder's tree is not \wf.
On the other hand, the honest responder replies to the adversarial challenger's queries truthfully.
Therefore, the adversarial challenger cannot pinpoint any violation.

The honest challenger runs Algorithm~\ref{alg.tree.vs.tree}, whereas the honest responder runs
Algorithm~\ref{alg.respond.bisection}.
The verifier forwards and verifies exactly up to $\log\ell$ inner node queries and one leaf query.
Then, at the end of the algorithm, the challenger arrives at the first point $j$ of disagreement,
and the honest responder reveals the leaf data $(\tx_j, \left<\st_{j}\right>)$
(Algorithm~\ref{alg.respond.bisection} Line~\ref{alg.respond.bisection.reveal}).
Finally, if $j \geq 1$, the honest responder also sends the leaf $(\tx_{j-1}, \left<\st_{j-1}\right>)$ at index $j-1$, along with its Merkle proof $\pi$
within its dirty tree in a single round of interaction (if $j = 0$, then the verifier already knows the contents of the first leaf).
This last response is only checked by the verifier and does not need to be forwarded to the challenger.
For brevity, we omit this portion from the responder's algorithm.

\import{./}{algorithms/alg.respond.bisection.tex}

If the responder is adversarial, she could send malformed responses.
We use the notation $\stc_j$ to denote the
\emph{claimed} $j^\text{th}$ state commitment by the responder,
but this may be malformed and does not necessarily correspond to an actual commitment
$\left<\st_j\right>$, where $\st_j$ is the $j^\text{th}$ state of an honest party's augmented dirty ledger.
In fact, it may not be a commitment at all. Similarly, the claimed tree root $\mroot^*$, provided by the
adversary may not necessarily be a correctly generated Merkle tree.

\import{./}{algorithms/alg.tournament.tex}

%% file: algorithms/alg.tree.vs.tree.tex
\begin{algorithm}
  \caption{\label{alg.tree.vs.tree}
    The algorithm ran by an honest challenger to identify the first point of disagreement against the responder's dirty tree, given that the trees have the same size $\ell$, but different roots. The variable $\challenger{\dtreesp}$ represents the challenger's dirty tree. 
   }
  \begin{algorithmic}[1]\small
    \Function{\sc \challenger{BisectionGame}}{\rm \challenger{\dtreesp}, \challenger{$\ell$}}
        \If{$\challenger{\ell} = 1$}
            \State\Return\Comment{We are done; let the verifier check the leaf}
        \EndIf
        \Receive{(\responder{h_l^*}, \responder{h_r^*})}{\responder{\textsc{Responder}}}\Comment{Ask to open inner node}\label{alg.tree.vs.tree.open-response}
        \Let{(\challenger{h_l}, \challenger{h_r})}{\challenger{\dtreesp\textrm{.left.root}}, \challenger{\dtreesp\textrm{.right.root}}}
        \If{$\challenger{h_l} = \responder{h_l^*}$}
            \Send{1}{\textsc{\responder{Responder}}}
            \State $\challenger{\textsc{BisectionGame}}(\challenger{\dtreesp\textrm{.right}}, \lfloor\frac{\challenger{\ell}}{2}\rfloor)$\label{alg.tree.vs.tree.recurse-right}
        \Else
            \Send{0}{\textsc{\responder{Responder}}}
            \State $\challenger{\textsc{BisectionGame}}(\challenger{\dtreesp\textrm{.left}}, \lfloor\frac{\challenger{\ell}}{2}\rfloor)$\label{alg.tree.vs.tree.recurse-left}
        \EndIf
    \EndFunction
    \end{algorithmic}
\end{algorithm}

%% file: algorithms/alg.respond.bisection.tex
\begin{algorithm}
  \caption{The algorithm ran during the bisection game by the responder to reply to the challenger's queries. The variable $\responder{\ledgeraug}$ denotes the responder's augmented dirty ledger. The algorithm $\textsc{MakeMerkleTree}$ returns the Merkle tree based on the given augmented dirty ledger.}
  \label{alg.respond.bisection}
  \begin{algorithmic}[1]\small
    \Function{\sc \responder{Respond}}{\responder{\ledgeraug}}
        \Let{\responder{\dtreesp^*}}{\textsc{MakeMerkleTree}(\responder{\ledgeraug})}
        \Send{\responder{\dtreesp^*}\textrm{.root}}{\challenger{\textsc{Challenger}}}
        \While{$\responder{\dtreesp^*}\textrm{.size} > 1$}
            \Let{(\responder{h_l^*, h_r^*})}{(\responder{\dtreesp^*}\textrm{.left.root}, \responder{\dtreesp^*}\textrm{.right.root})}
            \Send{(\responder{h_l^*}, \responder{h_r^*})}{\challenger{\textsc{Challenger}}}\label{alg.respond.bisection.open}
            \Receive{\challenger{\textrm{dir}}}{\challenger{\textsc{Challenger}}}
            \If{$\challenger{\textrm{dir}} = 0$}
                \Let{\responder{\dtreesp^*}}{\responder{\dtreesp^*}\textrm{.left}}
            \Else
                \Let{\responder{\dtreesp^*}}{\responder{\dtreesp^*}\textrm{.right}}
            \EndIf
        \EndWhile
        \Send{\responder{\dtreesp^*}\textrm{.data}}{\challenger{\textsc{Challenger}}}\label{alg.respond.bisection.reveal}
    \EndFunction
    \end{algorithmic}
\end{algorithm}

%% file: algorithms/alg.tournament.tex
\begin{algorithm}
  \caption{The tournament among the provers administered by the verifier. It takes a sequence of provers $\mathcal{P}$, ordered from the one with the largest alleged augmented dirty ledger size to the smallest. The algorithm $\textsc{Challenge}$ initiates a challenge game with the first given prover as the challenger and the second one as the responder.}
  \label{alg.tournament}

  \begin{algorithmic}[1]\small
    \Function{\sc Tournament}{$\mathcal{P}$}
      \Let{\textrm{sizes}}{\{\,\}}
      \For{$p \in \mathcal{P}$}
        \Let{\textrm{sizes}[p]}{p.\textrm{getsize()}}
      \EndFor
      \Let{\mathcal{S}}{\{\mathcal{P}[0]\}}
      \Let{\textrm{largest}}{\mathcal{P}[0]}
      \For{$i = 1 \text{ to } |\mathcal{P}| - 1$}
        \Do
          \If{$\textrm{largest}.\textrm{getsize()} > \textrm{sizes}[i]$}
          \Let{\textrm{result}}{\textsc{Challenge}(\textrm{largest}, \mathcal{P}[i])}\label{alg.tournament.protocol.1}
          \Else
          \Let{\textrm{result}}{\textsc{Challenge}(\mathcal{P}[i], \textrm{largest})}\label{alg.tournament.protocol.2}
          \EndIf
          \If{$\textrm{result} \text{ is ``nested MMRs''}$}          \Let{\mathcal{S}}{\mathcal{S} \cup \{\mathcal{P}[i]\}}\label{alg.tournament.both}
          \ElsIf{$\textrm{result} \text{ is ``\textrm{largest} loses''}$}
            \Let{\mathcal{S}}{\mathcal{S} \setminus \{\textrm{largest}\}}\label{alg.tournament.node}
            \Let{\textrm{largest}}{\argmax_{p \in \mathcal{S}}{\textrm{sizes}}}
          \EndIf
          \CommentLine{The set $\mathcal{S}$ is not updated if $\mathcal{P}[i]$ loses.}\label{alg.tournament.largest}
        \doWhile{$\textrm{result} \text{ is ``\textrm{largest} loses''} \land \mathcal{S} \neq \emptyset$}\label{alg.tournament.while}
        \If{$\mathcal{S} = \emptyset$}
          \Let{\mathcal{S}}{\{\mathcal{P}[i]\}}\label{alg.tournament.added}
        \EndIf
      \EndFor
      \State\Return{\textrm{largest}}
    \EndFunction
  \end{algorithmic}
\end{algorithm}

%% file: appendix_latency_tradeoff.tex
This section models the latency-bandwidth trade-off realized by tuning the degree $m$ of dirty trees, and complements the experimental results in Section~\ref{sec:implementation-results}. 
In an $m$-ary dirty tree representing
an $L$-sized ledger, the tree height decreases logarithmically
as the degree $m$ of the tree increases,
making the number of rounds of interactivity in the bisection game $\log_m{L}$.
However, the challenger must now indicate the index of the child to open,
making its messages $\log{m}$ bits in size. Similarly, when the responder opens up an inner node
and reveals its children, $m$ children need to be sent over the network.
If the hash used is $H$ bits long, then the messages sent by the responder
are $mH$ bits. There is therefore a \emph{latency/bandwidth} tradeoff in
the parameter $m$. A large $m$ incurs less interactivity, but larger network
messages, while a small $m$ incurs more interactivity but shorter
messages. In this section, we calculate the optimal $m$, given the respective
network parameters on bandwidth and latency.

Let $\Delta$ be the network latency between the prover and verifier,
measured in seconds, and $C$ be the communication bandwidth of the channels
connecting each prover to the verifier.
We assume that, upon downloading any given message, the prover and the verifier compute
the corresponding reply instantly (network latency dominates computational latency).
At each round of the bisection game, $\log{m}$ and $mH$ bits are downloaded by the responder and the challenger respectively.
Moreover, each message sent between the responder and challenger takes $2\Delta$ seconds to reach its destination,
because it has to be forwarded through the verifier.
Hence, each round is completed in $4\Delta+(mH+\log{m})/C$ seconds.
As the bisection game lasts for $\log_m{L}$ rounds, the total running time of the game becomes
$\left(4\Delta+mH/C+\log{m}/C\right)\log_m{L} = \frac{\log{L}}{\log{m}}\left(4\Delta+mH/C\right) + \log{L}/C$.
This expression is minimized for $m$ that satisfies the expression
$m(\log{m}-1) = 4\Delta C / H,$ \ie, $m=\exp{(W_1(4\Delta \frac{C}{eH}) + 1)}$, where $W_1$ is the Lambert $W$ function.
The different optimal $m$ for common bandwidths and latencies are plotted in Figure~\ref{fig:latency-bandwidth}.

\begin{figure}
    \centering
    \ifonecolumn
    \includegraphics[width=0.6\linewidth]{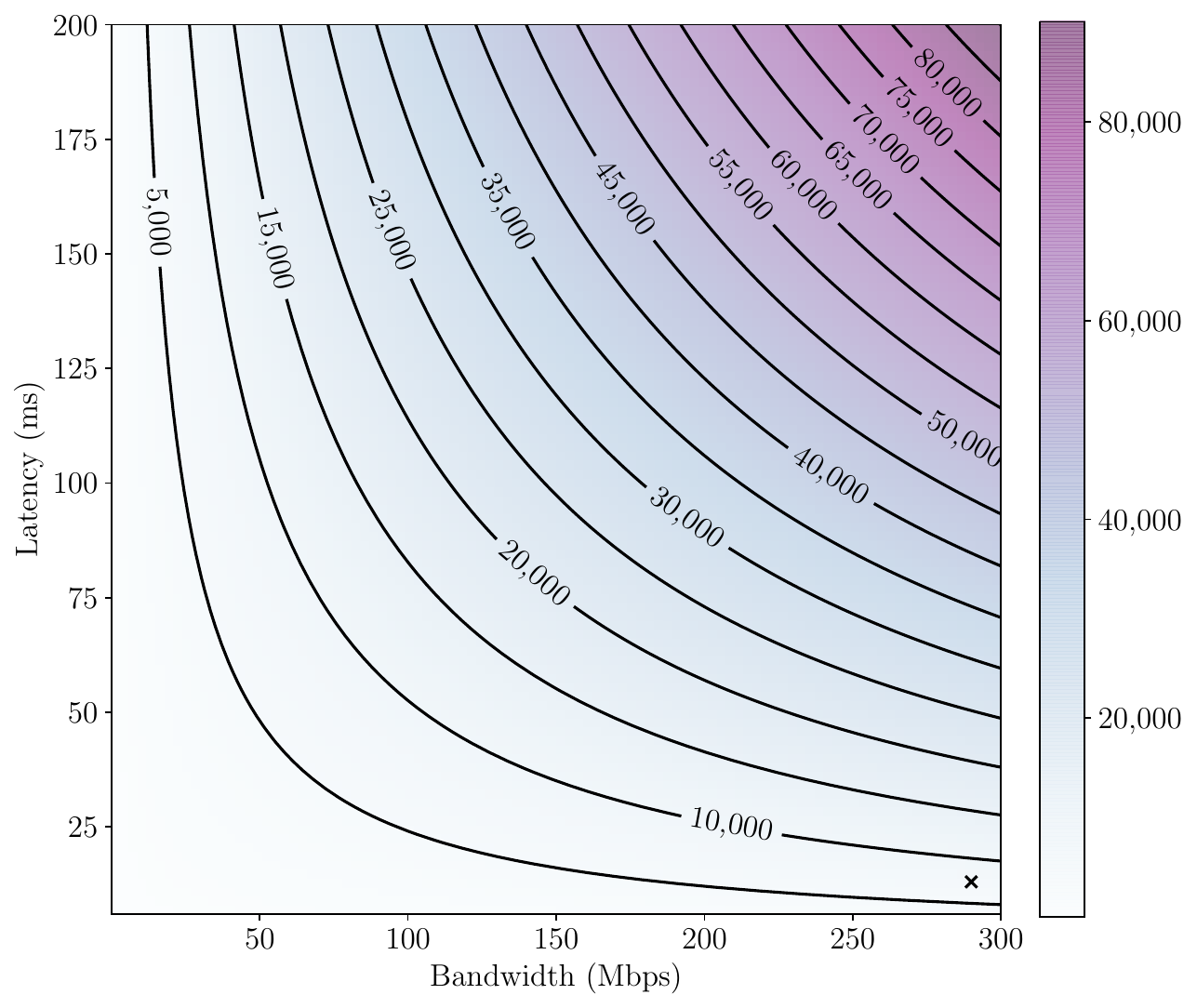}
    \fi
    \iftwocolumn
    \includegraphics[width=\linewidth]{figures/tradeoff.pdf}
    \fi
    \caption{Optimal Merkle tree degree $m$ (isolines) for a given network connection bandwidth (x-axis, in Mbps) and latency (y-axis, in ms). The $\times$ marker marks the particular example described in the text.}
    \label{fig:latency-bandwidth}
\end{figure}

Given\footnote{
\ifanon Typical conditions for the network connection in our university graduate student residences.
\fi
\ifnonanon
Typical conditions for the network connection in Stanford university graduate student residences.
\fi}
$C=290$ Mbps, $\Delta = 13$ ms, $H=256$ bits, the optimal $m$ is $7{,}442$,
yielding dirty trees that have quite a large degree. 
The optimal $m$ only depends
on the network parameters and not the ledger size.
Given the Ethereum ledger of $L = 1.5 \cdot 10^9$ transactions at the time of writing\footnote{Google Cloud Platform
BigQuery table \texttt{bigquery-{\allowbreak}public-{\allowbreak}data:{\allowbreak}crypto\_ethe\-reum.{\allowbreak}trans\-ac\-tions}
as of March 6th, 2022}, using the optimal $m$ for the given network parameters gives an estimate of $0.96$ seconds, where $0.86$ seconds of the time is due to the network delay $\Delta$.
This captures the duration of the whole bisection game with its $\log_m{L}$ rounds of interactivity.

%% file: super_light_clients.tex
In our construction, we abstracted the checking of transaction order that the verifier
performs into a consensus oracle, and discussed how this can be realized in the blockchain
setting using the standard SPV technique, achieving communication complexity of
$\mathcal{O}(C) = \mathcal{O}(r)$, where $C$ is the chain size and $r$ is the round
during which the light client is booting up. This gives a total of
$\mathcal{O}(C + \log L)$ communication complexity for our lazy light client protocol.
However, the consensus oracle can be replaced with a \emph{superlight} client that
does not download the whole header chain, and instead samples a small portion of it.
Such examples include interactive~\cite{popow} or non-interactive PoPoWs and PoPoS,
a primitive which
can be constructed using either superblocks~\cite{nipopows,logspace}, FlyClient~\cite{flyclient} or bisection games~\cite{popos}, and brings down the consensus oracle communication
complexity to a succinct $\mathcal{O}(\text{poly}\log C) = \mathcal{O}(\text{poly}\log r)$.
When composed with our protocol for identifying lazy ledger disagreements, the
total communication complexity then becomes
$\mathcal{O}(\text{poly}\log C + \log L) = \mathcal{O}(\text{poly} \log r)$,
which is the desirable succinctness. We highlight the different roles of each protocol
here: On the one hand, the superlight client, such as FlyClient or superblocks,
plays the role of the \emph{consensus oracle} and is used to answer queries about
\emph{which transaction succeeds another on the chain};
on the other hand, the interactive verification game is administered to determine
the current \emph{state of the world}, given access to such a consensus oracle.
The two protocols are orthogonal and can be composed to achieve an overall performant
system.

\iftwocolumn
\setlength{\tabcolsep}{6pt}
\renewcommand{\arraystretch}{1.5}
\begin{table*}[h]
\centering
\begin{tabular}{|p{2.5cm}|p{2.5cm}|p{2.5cm}|p{2.5cm}|p{2.5cm}|}
 \hline
  & \textbf{Full Node} & \textbf{Light Client} & \textbf{Superlight Client} & \textbf{Custodian Wallet} \\  [0.5ex]
 \hhline{|=|=|=|=|=|}
 \textbf{Interactivity} & $\mathcal{O}(1)$ & $\mathcal{O}(\log{L})$ & $\mathcal{O}(\log{L})$ & $\mathcal{O}(1)$ \\
 \hline
 \textbf{Communication} & $\mathcal{O}(C+L)$ & $\mathcal{O}(C+\log{L})$ & $\mathcal{O}(\log{C}+\log{L})$ & $\mathcal{O}(1)$ \\
 \hline
 \textbf{Decentralized} & $\checkmark$ & $\checkmark$ & $\checkmark$ & $\times$ \\
 \hline
\end{tabular}
\vspace{2pt}
\caption{Comparison of different client types on a lazy blockchain.}
\label{tab:table-comp}
\end{table*}
\fi

\ifonecolumn
\setlength{\tabcolsep}{6pt}
\renewcommand{\arraystretch}{1.5}
\begin{table*}[ht]
\centering
\begin{tabular}{|p{2.5cm}|p{1.7cm}|p{2cm}|p{1.9cm}|p{2cm}|}
 \hline
  & \textbf{Full} \newline \textbf{Node} & \textbf{Light} \newline \textbf{Client} & \textbf{Superlight} \newline \textbf{Client} & \textbf{Custodian} \newline \textbf{Wallet} \\  [0.5ex]
 \hhline{|=|=|=|=|=|}
 \textbf{Interactivity} & $\mathcal{O}(1)$ & $\mathcal{O}(\log{L})$ & $\mathcal{O}(\log{L})$ & $\mathcal{O}(1)$ \\
 \hline
 \textbf{Communication} & $\mathcal{O}(C+L)$ & $\mathcal{O}(C+\log{L})$ & $\mathcal{O}(\log{CL})$ & $\mathcal{O}(1)$ \\
 \hline
 \textbf{Decentralized} & $\checkmark$ & $\checkmark$ & $\checkmark$ & $\times$ \\
 \hline
\end{tabular}
\vspace{2pt}
\caption{Comparison of different client types on a lazy blockchain.}
\label{tab:table-comp}
\end{table*}
\fi

The interactivity and communication complexity
for synchronization times for lazy light clients composed with different
consensus oracles is illustrated in Table~\ref{tab:table-comp}.
A \emph{Full Node} (left-most column) downloads the whole header chain of size $C$
and every transaction of size $L$, thus does not need to play any
interactive games, achieving constant interactivity but large communication
complexity. A \emph{Custodian Node} (right-most column) is a wallet that
trusts a server to deliver correct data and does not verify it (e.g., MetaMask);
this has the best performance in both complexity and interactivity.
These were the only two previously known means of constructing clients for
lazy blockchains.
The two protocols titled \emph{Light Client} and \emph{Superlight Client}
in the middle columns
are clients composed with the lazy light clients explored in this
work. In the light client case, an SPV
client is used for the consensus oracle, while in the super light client
case, a NIPoPoW superblock client is used for the consensus oracle. The
$C$ or $\log C$ term stems from the underlying consensus oracle, while
the $\log L$ term stems from our lazy protocol.

%% file: generalize.tex
So far, we assumed the underlying consensus protocol is blockchain-based and the computation of state
mimicks Ethereum's EVM. There exist different architectures for consensus (\eg, DAG-based
constructions~\cite{parallel-chains,phantom,spectre}) and execution (\eg, UTXO~\cite{bitcoin-dev-guide}).
Our protocol is generic and agnostic to these details.
To enable the formal analysis of our construction in a generic manner,
in this section we axiomatize the protocol requirements regarding consensus and execution.

\iflong
\begingroup
\setlength{\tabcolsep}{6pt} %
\renewcommand{\arraystretch}{1.5} %
\begin{table*}[]
    \centering
    \begin{tabular}{|l|l|l|l|}
    \hline
    \textbf{Primitive}                  & \textbf{Axiom}        & \textbf{Param} & \textbf{Requirement}                                                                                    \\ \hhline{|=|=|=|=|}
    \multirow{3}{*}{\textbf{Ledger}}    & \emph{Safety}       & -                  & $\ledger^{P_1}_{r_1} \preceq \ledger^{P_2}_{r_2}$                                                       \\ \cline{2-4}
                                        & \emph{Liveness}     & $u$                & $\tx$ in $\ledger^{P}_{r+u}$                                                                            \\ \cline{2-4}
                                        & \emph{Lipschitz}    & $\alpha$           & $|\ledger^P_{r_2}| - |\ledger^P_{r_1}| \leq \alpha(r_2 - r_1)$                                          \\ \hhline{|=|=|=|=|}
    \multirow{3}{*}{\textbf{Consensus}}
                                        & \emph{Completeness} & -                  & $(\tx, \tx')$ in $\ledger^{\cup}_{r} \rightarrow \mathcal{CO}_r(\tx, \tx')$                               \\ \cline{2-4}
                                        & \emph{Soundness}    & $\nu$           & $(\tx, \tx')$ not in $\ledger^{\cup}_{r + \nu} \rightarrow \lnot \mathcal{CO}_r(\tx, \tx')$              \\ \cline{2-4}
                                        & \emph{Succinctness} & $f$                & $f(r) \in \mathcal{O}(\text{poly} \log r)$                                                              \\ \hhline{|=|=|=|=|}
    \multirow{3}{*}{\textbf{Execution}}
                                        & \emph{Completeness} & -                  & $\transition(\st, \tx) = \st' \rightarrow \left<\transition\right>(\stc, \tx, \pi) = \left<\st'\right>$ \\ \cline{2-4}
                                        & \emph{Soundness}    & -                  & $\left<\transition(\st, \tx)\right> \neq \left<\transition\right>(\stc, \tx, \pi)$ is hard              \\ \cline{2-4}
                                        & \emph{Succinctness} & $g$                & $g(r) \in \mathcal{O}(\text{poly}\log r)$                                                               \\ \hline
    \end{tabular}
    \caption{The $9$ axioms required to construct a succinct light client.}
    \label{tab:axioms}
\end{table*}
\endgroup
\fi

\subsection{Consensus Protocol}
\label{sec:consensus-protocol}
We consider a consensus protocol $\mathcal{C}$ executed by honest full nodes.
Each honest full node $P$ exposes a \emph{read ledger} functionality which, at round $r$,
returns a finite sequence $\LOGdirty{r}{P}$ of stable transactions as the dirty ledger.
They also expose a \emph{write transaction} functionality which, given a transaction $\tx$
at round $r$, attempts to include the transaction into the ledger.

Our consensus protocol must satisfy the following properties.

\begin{definition}[Ledger Safety]
A consensus protocol $\mathcal{C}$ is \emph{safe} if for all honest parties $P_1, P_2$
at rounds $r_1 < r_2$, $\LOGdirty{r_1}{P_1}$ is a prefix of $\LOGdirty{r_2}{P_2}$.
\end{definition}

\begin{definition}[Ledger Liveness]
A consensus protocol $\mathcal{C}$ is \emph{live} with liveness parameter $u$ if,
when an honest party $P$ attempts to
write a transaction $\tx$ to the ledger at round $r$, the transaction appears
in $\LOGdirty{r + u}{P}$.
\end{definition}

Proof-of-work (in static and variable difficulty)
and proof-of-stake protocols satisfy the above properties~\cite{backbone,varbackbone,ouroboros,praos,snowwhite}.

\begin{definition}[Ledger Lipschitz]
A consensus protocol $\mathcal{C}$ is \emph{Lipschitz} with parameter $\alpha$
if for every honest party $P$
and rounds $r_1 \leq r_2$, we have that
$|\LOGdirty{r_2}{P}| - |\LOGdirty{r_1}{P}| < \alpha(r_2 - r_1)$.
\end{definition}

The above requirement states that ledgers grow at a bounded rate. 
For protocols that have a longest chain component such as Prism, Snap-and-Chat, Ouroboros, Babylon and Bitcoin, this follows from the fact that chains have an upper bound in their growth rate~\cite[Lemma 13]{backbone},
and that the number of transactions in each block is limited by a constant.
Such a chain growth upper bound is also present for Tendermint, which Celestia is based on, as each Tendermint round has a duration of at least $\Delta$~\cite{tendermint}.
Given the definitions of $\alpha$, $u$ and $\nu$ in the section below, we express the parameter $\psi$ used in the suffix monologue as $\psi = \alpha(u+\nu)$.

\subsection{\COracle}
\label{sec:consensus-oracle}

To enable queries about the order of transactions on the dirty ledger, we assume that the lazy blockchain
protocol provides access to a \emph{\coracle} $\mathcal{CO}$. The \coracle is a single-round black box interactive protocol executed among
the verifier and the provers. The verifier invokes the oracle with input two transactions $(\tx, \tx')$ and
receives a boolean response.
The goal of the verifier is to determine whether a transaction $\tx'$ \emph{immediately follows} another transaction $\tx$ on $\ledgercup$
(\ie, $(\tx, \tx') \mid \ledgercup$).

We require that the consensus oracle satisfies the following properties.

\begin{definition}[Consensus Oracle Security]
\label{def:coracle-p}
An \coracle is \emph{secure} if it satisfies:
\begin{itemize}
    \item \textbf{Completeness.}
    $(\tx, \tx')\mid \ledgercup_r \Rightarrow \mathcal{CO}_r(\tx,\tx')$.
    \item \textbf{Soundness.} The \coracle is sound with \emph{delay parameter} $\nu$ if for any PPT adversary $\mathcal{A}$,
    $
    \Pr[(\tx, \tx', r) \gets \mathcal{A}(1^\lambda);\allowbreak
        \mathcal{CO}_r(\tx,\tx') \land (\tx, \tx') \nmid \ledgercup_{r + \nu}] \leq \negl(\lambda)\,.
    $
\end{itemize}
\end{definition}

\begin{definition}[Consensus Oracle Communication Complexity]
A \coracle has communication complexity $f(r)$ if the total size of the query and response messages exchanged during an oracle query invoked at round $r$
is $f(r) \in \mathcal{O}(r)$.
\end{definition}

We assume that every transaction in the dirty ledger is unique and there are no duplicate transactions.
Under this assumption, a \coracle on top of the Nakamoto longest chain consensus protocol can be instantiated as follows.
Since the construction below also applies to protocols that output a chain of blocks, we will refer to the longest chain as the \emph{canonical} chain.

The blockchain consists of a header chain, each header containing the Merkle tree root, \ie the transaction root, of the transactions organized within the associated block.
The ordering of the blocks by the header chain together with the ordering of the transactions by each Merkle tree determine the total order across all transactions.
Thus, to query the \coracle with the two transactions $\tx$ and $\tx' \neq \tx$, the verifier first downloads all the block headers
from the honest provers and determines the canonical stable header chain.
Then, it asks a prover if $\tx$ immediately precedes $\tx'$ on its dirty ledger.
To affirm, the prover replies with
(a) the positions $i$ and $i'$ of the transactions $\tx$ and $\tx'$ within their respective Merkle trees,
(b) the Merkle proofs $\pi$ and $\pi'$ from the transactions $\tx$ and $\tx'$ to the transaction roots,
(c) the positions $j \leq j'$ of the block headers containing these transaction roots, on the canonical stable header chain.

Then, the verifier checks that the Merkle proofs are valid, and accepts the prover's claim iff either of
(1) the two blocks are the same, \ie, $j'=j$, and $i'=i+1$, or
(2) otherwise, the two blocks are consecutive, \ie, $j'=j+1$, and $i$ is the index of the last leaf in the tree of block $j$ while $i'$ is the index of the first leaf in the tree of block $j' = j + 1$.

If no prover is able to provide such a proof, the oracle returns \emph{false} to the verifier. The oracle's soundness follows the ledger safety.

The above is one example instantiation of a \coracle.
\iflong
Appendix~\ref{sec:coracle-constructions} gives proofs of completeness and soundness for the \coracle
as well as notes on how it can be implemented on different blockchain protocols.
\else
Proofs of completeness and soundness for the \coracle are available in the full version of this paper.
\fi

\iflong
To relax the uniqueness assumption for the transactions in the dirty ledger, each augmented dirty ledger entry containing a transaction $\tx$ can be extended by adding the index $j_{\tx}$ of the header of the block containing $\tx$, and the index of $\tx$ within the Merkle tree of that block.
In this case, the verifier queries the \coracle not only with transactions $\tx$ and $\tx' \neq \tx$, but also with the corresponding block header and Merkle tree indices $j_{\tx}$, $i_{\tx}$ and $j_{\tx'}$, $i_{\tx'}$.
Hence, during the query, the verifier also checks if the block and transaction indices for $\tx$ and $\tx'$, \eg, $j,i$ and $j',i'$, received from the prover matches the claimed indices: $j=j_{\tx}$, $j'=j_{\tx'}$, $i=i_{\tx}$, $i'=i_{\tx'}$.
\fi

\subsection{\EOracle}
\label{sec:execution-oracle}

To enable queries about the validity of state execution, we assume that the lazy blockchain
protocol provides access to an \emph{\eoracle}. The \eoracle is a single-round black box interactive protocol executed among
the verifier and the provers. The verifier invokes the oracle with a transaction $\tx$ and two state commitments, $\stc$ and $\stc'$ as input, and
receives a boolean response.
The goal of the verifier is to determine whether there exists a state $\st$ such that $\stc$ is the commitment of $\st$ and $\stc' = \left<\transition(\tx,\st)\right>$.

The \eoracle is parametrized by a triplet $(\transition, \left<\cdot\right>, \left<\transition\right>)$
consisting of an efficiently computable
\emph{transition function} $\transition(\cdot, \cdot)$,
a \emph{commitment scheme} $\left<\cdot\right>$,
and
a \emph{succinct transition function} $\left<\transition\right>(\cdot, \cdot, \cdot)$.
The succinct transition function $\left<\transition\right>$ accepts a state commitment $\stc$,
a transaction $\tx$, and a proof $\pi$, and produces
a new state commitment $\stc'$ which corresponds to the commitment of the updated state.

To query the \eoracle on $\tx$, $\stc$ and $\stc'$, the verifier first asks a prover for a proof $\pi$.
If the prover claims that he knows a state $\st$ such that $\stc$ is the commitment of $\st$ and $\stc' = \left<\transition(\tx,\st)\right>$, it gives a proof $\pi$.
Then, the verifier accepts the prover's claim if $\stc' = \left<\transition\right>(\stc, \tx, \pi)$.
Otherwise, if $\left<\transition\right>$ throws an error or outputs a different commitment, the verifier rejects the claim.

\begin{definition}[Execution Oracle Security]
\label{def:eoracle-p}
An \eoracle is \emph{secure} if it satisfies:

\noindent
\textbf{Completeness.} \Eoracle is \emph{complete} with respect to a proof-computing PPT machine $M$ if
for any state $\st$ and transaction $\tx$, it holds that $M(\st, \tx)$ outputs a $\pi$ that satisfies
$\left<\delta\right>(\left<st\right>, \tx, \pi) = \left<\delta(\st, \tx)\right>$.

\noindent
\textbf{Soundness.}
For any PPT adversary $\mathcal{A}$:
\begin{align*}
\Pr[(\st, \tx, \pi) \gets \mathcal{A}(1^\lambda); \left<\delta(\st, \tx)\right> \neq \left<\delta\right>(\left<\st\right>, \tx, \pi)]
\leq \negl(\lambda)\,.
\end{align*}
\end{definition}

\begin{definition}[Execution Oracle Communication Complexity]
An \eoracle has communication complexity $g(r)$ if the total size of the query and response messages exchanged during an oracle query invoked at round $r$
is $g(r) \in \mathcal{O}(r)$.
\end{definition}

In the account based model~\cite{albassam2018fraud}, the state is a Sparse Merkle Tree (SMT)~\cite{sparse-mt}
representing a key-value store.
The values constitute the leaves of the SMT and the keys denote their indices.
The state commitment corresponds to the root.

The verifier queries the \eoracle with $\tx$, $\stc$ and $\stc'$.
Suppose there is a state $\st$ with commitment $\stc$ and $\stc' = \left<\transition(\st, \tx)\right>$.
Let $D$ denote the leaves of the SMT $\st$. Let $\mathcal{S}_\tx$ be the keys of the SMT that the transaction $\tx$ reads from or writes to\footnote{In Ethereum, these can be obtained by the verifier via the $\mathsf{eth\_createAccessList}$ RPC.}.
We assume that the number of leaves touched by a particular transaction is constant.
Then, the proof required by $\left<\delta\right>$ consists of:
\begin{itemize}
    \item The key-value pairs $(i, D[i])$ for $i \in \mathcal{S}_{\tx}$ within $\st$.
    \item The Merkle proofs $\pi_i$, $i \in \mathcal{S}_{\tx}$, from the leaves $D[i]$ to the root $\stc$.
\end{itemize}

Given the components above, $\left<\delta\right>$ verifies the proofs $\pi_i$ and the validity of $\tx$ with respect
to the pairs $(i,D[i])$, \eg, $\tx$ should not be spending from an account with zero balance.
If there are pairs read or modified by $\tx$ that have not been provided by the prover, then $\left<\delta\right>$ outputs $\bot$.
If all such key-value pairs are present and $\tx$ is invalid with respect to them, $\left<\delta\right>$ outputs $\stc$, and does not modify the state commitment.
Otherwise, $\left<\delta\right>$ modifies the relevant key-value pairs covered by $\mathcal{S}_{\tx}$, which can be done efficiently~\cite{practical-delegation}.
Finally, it calculates the new SMT root, \ie the new state commitment,
using the modified leaves and the corresponding Merkle proofs among $\pi_i$, $i \in \mathcal{S}_{\tx}$.

SMTs can also be used to represent states based on the UTXO~\cite{bitcoin} model.
In this case, the value at each leaf of the SMT is a UTXO.
Thus, the execution oracle construction above generalizes to the UTXO model.

%% file: appendix_consensus_oracles.tex
\iflong
\Coracle constructions for Celestia (LazyLedger)~\cite{lazyledger}, Prism~\cite{prism}, and Snap-and-Chat~\cite{ebb-and-flow,snap-and-chat} follow the same paradigm described in Section~\ref{sec:consensus-oracle}.
\fi

\subsection{Celestia}

Celestia uses Tendermint~\cite{tendermint} as its consensus protocol, which outputs a chain of blocks containing transactions.
Blocks organize the transactions as namespaced Merkle trees, and the root of the tree is included within the block header.
Hence, the construction of Section~\ref{sec:consensus-oracle} can be used to provide a \coracle for Celestia.

Celestia is designed as a data availability and consensus layer for multiple rollups. 
However, as Celestia is a lazy blockchain, each rollup on Celestia (called ‘sovereign rollups’) also need a mechanism for their \emph{rollup light clients} to discover the correct latest rollup state.
Towards this goal, our succinct light client construction can be utilized by the rollup nodes to support these light clients. 
For instance, as rollups are maintained by full nodes that execute the rollup-specific transactions (ignoring other transactions) posted to Celestia, these nodes can aid the rollup light clients by creating a dirty ledger of rollup-specific transactions, the corresponding dirty trees and MMRs, in the same way as the full nodes of a lazy blockchain with a single state transition function would help its light clients discover the correct latest state.

\subsection{Prism}

In Prism, a transaction $\tx$ is first included within a transaction block.
This block is, in turn, referred by a proposer block.
Once the proposer block is confirmed in the view of a prover $\prover$ at round $r$, $\tx$ enters the ledger $\LOGdirty{r}{\prover}$.
Hence, the proof of inclusion for $\tx$ consists of two proofs: one for the inclusion of $\tx$ in a transaction block $B_T$, the other for the inclusion of the header of $B_T$ in a proposer block $B_P$.
If transactions and transaction blocks are organized as Merkle trees, then, the proof of inclusion for $\tx$ would be two Merkle proofs: one from $\tx$ to the Merkle root in the header of $B_T$, the other from the header of $B_T$ to the Merkle root in the header of $B_P$.

The construction of Section~\ref{sec:consensus-oracle} can be generalized to provide a \coracle for Prism.
In this case, to query the \coracle with two transactions $\tx$ and $\tx' \neq \tx$, the verifier first downloads all the proposal block headers from the honest provers and determines the longest stable header chain.
Then, it asks a prover if $\tx$ immediately precedes $\tx'$ on its dirty ledger.

To affirm, the prover replies with:
\begin{itemize}
    \item the positions $i_t$ and $i'_t$ of the transactions $\tx$ and $\tx'$ within their respective Merkle trees contained in the respective transaction blocks $B_T$ and $B'_T$.
    \item the Merkle proofs $\pi_t$ and $\pi'_t$ from the transactions $\tx$ and $\tx'$ to the corresponding Merkle roots within the headers of $B_T$ and $B'_T$,
    \item the positions $i_p$ and $i'_p$ of the headers of the transaction blocks $B_T$ and $B'_T$ within their respective Merkle trees contained in the respective proposal blocks $B_P$ and $B'_P$.
    \item the Merkle proofs $\pi_p$ and $\pi'_p$ from the headers of $B_T$ and $B'_T$ to the corresponding Merkle roots within the headers of $B_P$ and $B'_P$,
    \item the positions $j \leq j'$ of the headers of $B_P$ and $B'_P$ on the longest stable header chain.
\end{itemize}
Then, the verifier checks that the Merkle proofs are valid, and accepts the prover's claim if and only if either of the following cases hold:
\begin{enumerate}
\item if $j=j'$ and $i'_p=i_p$, then $i'_t=i_t+1$.
\item if $j=j'$ and $i'_p=i_p+1$, then $i_t$ is the index of the last leaf in the tree of $B_T$ while $i'_t$ is the index of the first leaf in the tree of $B'_T$.
\item if $j'>j$, then $i_p$ is the index of the last leaf in the tree of $B_P$ while $i'_p$ is the index of the first leaf in the tree of $B'_P$. Similarly, $i_t$ is the index of the last leaf in the tree of $B_T$ while $i'_t$ is the index of the first leaf in the tree of $B'_T$.
\end{enumerate}
If no prover is able to provide such a proof, the oracle returns \emph{false} to the verifier.

\subsection{Snap-and-Chat}

In Snap-and-Chat protocols, a transaction $\tx$ is first included within a block $B_T$ proposed in the context of a longest chain protocol.
Upon becoming $k$-deep within the longest chain, where $k$ is a predetermined parameter, $B_T$ is, in turn, included within a block $B_P$ proposed as part of a partially-synchronous BFT protocol.
Once $B_P$ is finalized in the view of a prover $\prover$ at round $r$, $\tx$ enters the \emph{finalized ledger} $\LOGdirty{r}{\prover}$.
There are again two Merkle proofs for verifiable inclusion, one from $\tx$ to the Merkle root included in $B_T$, the other from the header of $B_T$ to the Merkle root in the header of $B_P$.
Consequently, the \coracle construction for Prism also applies to Snap-and-Chat protocols.

\begin{remark}
We remark here that protocols that do not follow the longest chain rule, or are not even proper chains, can be utilized by our protocol.
Such examples include Parallel Chains~\cite{parallel-chains}, PHANTOM~\cite{phantom} / GHOSTDAG~\cite{ghostdag}, and SPECTRE~\cite{spectre}.
The only requirement is that these systems provide a succinct means of determining whether two transactions follow one another on the ledger.
\end{remark}

\subsection{Colored Coins and Babylon}
Colored coins~\cite{colored-coin} refer to assets other than Bitcoin that are maintained on the Bitcoin blockchain, and derive their security from the consensus security of Bitcoin.
Babylon~\cite{babylon} is a protocol that checkpoints off-the-shelf PoS protocols onto Bitcoin to mitigate PoS-related problems such as non-slashable posterior corruption attacks, low liveness resilience and difficulty to bootstrap from low token valuation.
To post checkpoints and other types of data, Babylon and other colored coin applications use the $\mathsf{OP\_RETURN}$ scripting code, which allows arbitrary data to be recorded in an unspendable Bitcoin transaction.
Since the miners do not check the validity of the data within the $\mathsf{OP\_RETURN}$ transactions with respect to the corresponding application state, Bitcoin acts as a lazy blockchain towards these applications.
Section~\ref{sec:coracle-longest-chain} describes how a \coracle can be instantiated for longest chain protocols such as Bitcoin.

In the case of Babylon, provers can interact with an \coracle after the bisection game to prove the validity of the checkpoint at the first point of disagreement.
To support an \coracle, each checkpoint posted to Bitcoin must be augmented by the \emph{active validator set} of the portion of the PoS protocol ledger corresponding to the checkpoint (\cf \cite{babylon}[Sections IV-C and V-B]).
To verify the validity of the disputed checkpoint (via \cite{babylon}[Algorithms 1 and 2]), the verifier has to read only a constant-size portion of the PoS protocol ledger, namely the portion between the chain checkpointed by the earlier, common checkpoint, and the latter checkpoint at the source of disagreement.

\subsection{Longest Chain}
\label{sec:coracle-longest-chain}
As an illustration of how the consensus oracle can be realized in a longest header chain protocol, we provide sketches for the proofs of \coracle completeness and soundness in the Nakamoto setting.

Even in the original Nakamoto paper~\cite{bitcoin}, a description of an SPV client is provided, and it realizes our consensus oracle axioms, although these virtues were not stated or proven formally. The consensus oracle works as follows. The verifier connects to multiple provers, at least one of which is assumed to be honest. It inquires of the provers their longest chains, downloads them, verifies that they are chains and that they have valid proof-of-work, and keeps the heaviest chain. It then chops off $k$ blocks from the end to arrive at the stable part. Upon being queried on two transactions $(\tx, \tx')$, the oracle inquires of its provers whether these transactions follow one another on the chain. To prove that they do, the honest prover reveals two Merkle proofs of inclusion for $\tx$ and $\tx'$. These must appear in either consecutive positions within the same block header, or at the last and first position in consecutive blocks.

The terminology of \emph{typical executions}, the Common Prefix parameter $k$, and the Chain Growth parameter $\tau$
are borrowed from the Bitcoin Backbone~\cite{backbone} line of works, where these properties are proven. We leverage
these properties to show that our Consensus Oracle satisfies our desired axioms. Our proofs are in the static synchronous
setting, but generalize to the $\Delta$-bounded delay and variable difficulty settings.

\begin{lemma}[Nakamoto Completeness]
  In typical executions where honest majority is observed,
  the Nakamoto Consensus Oracle is complete.
\end{lemma}
\begin{proof}[Sketch]
  We prove that, if $(\tx, \tx')$ are reported in $\ledger^{\cup}$, then an honest prover will be able to prove so.
  Suppose the verifier chose a longest header chain $C^V$.
  If $(\tx, \tx')$ appear consecutively in $\ledger^{\cup}$, by ledger safety, this means that they belong to the ledger of at least one honest
  party $P$ who is acting as a prover. Since $(\tx, \tx')$ appear consecutively in $\ledger^P$, therefore they appear in the stable portion
  $C^P[{:}-k]$ of the chain $C^P$ held by $P$. By the Common Prefix property, $C^P[{:}-k]$ is a prefix of $C^V$
  and therefore $(\tx, \tx')$ appear consecutively in the stable header chain adopted by the verifier.
  Therefore, the verifier accepts.
\end{proof}

\begin{lemma}[Nakamoto Soundness]
  In typical executions where honest majority is observed,
  the Nakamoto Consensus Oracle instantiated with a Merkle Tree that uses a collision resistant
  hash function is sound, with soundness parameter $\nu = \frac{k}{\tau}$ where
  $k$ is the Common Prefix parameter and $\tau$ is the Chain Growth parameter.
\end{lemma}
\begin{proof}[Sketch]
  Suppose for contradiction that $(\tx, \tx')$ are not reported in $\ledger^{\cup}_{r + \nu}$, yet the adversary
  convinces the verifier of this at round $r$. This means that the adversary has presented some header chain $C^V$
  to the verifier which was deemed to be the longest at the time, and $(\tx, \tx')$ appear in its stable portion
  $C^V[{:}-k]$. Consider an honest prover $P$. At time $r$, the honest prover holds a chain $C^P_r$ and at round
  $r + \nu$, it holds a chain $C^P_{r + \nu}$.
  By the Common Prefix property, $C^V[{:}-k]$ is a prefix of $C^P_r$
  and of $C^P_{r + \nu}$. Furthermore, $(\tx, \tx')$ will appear in the same block (or consecutive blocks)
  in all three.
  By the Chain Growth property, $C^P_{r + \nu}$ contains at least $k$ more blocks than
  $C^P_r$. Therefore, $(\tx, \tx')$ appears in $C^P_{r + \nu}[{:}-k]$ and are
  part of the stable chain at round $r + \nu$ for party $P$. They are hence reported
  in $\ledger^P_{r + \nu} \subseteq \ledger^{\cup}_{r + \nu}$, which is a contradiction.
  Finally, by Proposition~\ref{prop:adv-merkle}, proofs of inclusion for $\tx$ and $\tx'$ cannot be forged with respect to Merkle roots in block headers other than those in $C^V$, that were initially shown to contain $\tx$ and $\tx'$ (except with negligible probability).
\end{proof}

Proofs of correctness and soundness for the \coracle constructions of Celestia, Prism and Snap-and-Chat follow a similar pattern to the proofs for the Nakamoto setting.

Lastly, for succinctness, one must leverage a construction such as superblock NIPoPoWs~\cite{nipopows}.
Here, proofs of the longest chain are $\text{poly}\log C$ where $C$ denotes the chain size.
Transaction inclusion proofs make use of \emph{infix proofs}~\cite{nipopows} in addition to
Merkle Tree proofs of inclusion into block headers. As $C \in \mathcal{O}(\text{poly}\log r)$,
these protocols are $\mathcal{O}(\text{poly}\log r)$ as desired. Completeness and soundness
follow from the relevant security proofs of the construction.

%% file: appendix_proofs.tex
\iflong
Our proof structure is as follows. First, we prove some facts about the bisection game,
in particular its succinctness, soundness, and completeness. We later leverage these results
to show that our full game enjoys the same virtues.
This section is based on the generalized model for lazy blockchains presented in Section~\ref{sec:generalize}, and the axioms used by the proofs are given by Table~\ref{tab:axioms}.
\fi

\ifshort
\begingroup
\renewcommand{\arraystretch}{1.5} %
\begin{table*}[]
    \centering
    \begin{tabular}{|l|l|l|}
    \hline
    \textbf{Primitive}                  & \textbf{Axiom}         & \textbf{Requirement}                                                                                    \\ \hhline{|=|=|=|}
    \multirow{3}{*}{\textbf{Ledger}}    & \emph{Safety}          & $\ledger^{P_1}_{r_1} \preceq \ledger^{P_2}_{r_2}$                                                       \\ \cline{2-3}
                                        & \emph{Liveness}        & $\tx$ in $\ledger^{P}_{r+u}$                                                                            \\ \cline{2-3}
                                        & \emph{Lipschitz}       & $|\ledger^P_{r_2}| - |\ledger^P_{r_1}| \leq \alpha(r_2 - r_1)$                                          \\ \hhline{|=|=|=|}
    \multirow{3}{*}{\textbf{Consensus}}
                                        & \emph{Completeness}    & $(\tx, \tx')$ in $\ledger^{\cup}_{r} \rightarrow \mathcal{CO}_r(\tx, \tx')$                               \\ \cline{2-3}
                                        & \emph{Soundness}       & $(\tx, \tx')$ not in $\ledger^{\cup}_{r + \nu} \rightarrow \lnot \mathcal{CO}_r(\tx, \tx')$              \\ \cline{2-3}
                                        & \emph{Succinctness}    & $f(r) \in \mathcal{O}(\text{poly} \log r)$                                                              \\ \hhline{|=|=|=|}
    \multirow{3}{*}{\textbf{Execution}}
                                        & \emph{Completeness}    & $\transition(\st, \tx) = \st' \rightarrow \left<\transition\right>(\stc, \tx, \pi) = \left<\st'\right>$ \\ \cline{2-3}
                                        & \emph{Soundness}       & $\left<\transition(\st, \tx)\right> \neq \left<\transition\right>(\stc, \tx, \pi)$ is hard              \\ \cline{2-3}
                                        & \emph{Succinctness}    & $g(r) \in \mathcal{O}(\text{poly}\log r)$                                                               \\ \hline
    \end{tabular}
    \caption{The $9$ axioms required to construct a succinct light client.}
    \label{tab:axioms}
\end{table*}
\endgroup
\fi

\ifshort
The proofs of the following two lemmas are straightforward and are provided in the full version.
\fi

\begin{restatable}[Bisection Succinctness]{lemma}{restateBisectionSuccinctness}
\label{lem:bisection-succinctness}
Consider a \coracle and an \eoracle with $f$ and $g$ communication complexity respectively.
Then, the bisection game invoked at round $r$ with trees of size $\ell$ ends in $\log(\ell)$ rounds of communication and
has a total communication complexity of $O(\log{\ell}+f(r)+g(r))$.
\end{restatable}
\iflong
\begin{proof}
When the dirty trees have $\ell$ leaves, there can be at most $\log{\ell}$ valid queries,
as the verifier aborts the game after $\log\ell$ queries.
Hence, the bisection game ends in $\log{\ell}$ rounds of interactivity.

At each round of communication, the challenger indicates whether he wants the left or the right child to be opened (which can be designated by a constant number of bits), and the responder replies with two constant size hash values.
At the final round, the responder returns $(\tx_{j-1},\stc_{j-1})$ and $(\tx_j,\stc_j)$, the augmented dirty ledger entries at indices $j-1$ and $j$, along with the Merkle proof for the $j-1^{\text{st}}$ entry (Alternatively, it only returns $(\tx_0,\stc_0)$).
The entries have constant size since transactions and state commitments are assumed to have constant sizes.
The Merkle proof consists of $\log{\ell}$ constant size hash values.
Consequently, the total communication complexity of the bisection game prior to the oracle queries becomes $O(\log{\ell})$.

Finally, the verifier queries the \coracle on $(\tx_{j-1},\tx_j)$ and the \eoracle on $(\stc_{j-1},\tx_j,\stc_j)$ with $O(f(r))$ and $O(g(r))$ communication complexity.
Hence, the total communication complexity of the bisection game becomes $O(\log(\ell)+f(r)+g(r))$.
\end{proof}
\fi

\begin{restatable}[Bisection Completeness]{lemma}{restateBisectionCompleteness}
\label{lem:bisection-completeness}
Suppose the consensus and execution oracles are complete and the ledger is safe.
Then, the honest responder wins the bisection game against any PPT adversarial challenger.
\end{restatable}

\iflong
\begin{proof}
We will enumerate the conditions checked by the verifier in Algorithm~\ref{alg.verifier} to show that the honest responder always wins.

The honest responder replies to each valid query from the verifier, and the replies are syntactically valid.
Hence, conditions (1) and (2) of Algorithm~\ref{alg.verifier} cannot fail.

By the construction of the honest responder's Merkle tree, each inner node $h$ queried by the challenger satisfies $h=H(h_l \concat h_r)$ for its children $h_l, h_r$ returned
in response to the query.
For the same reason, the Merkle proof given by the responder is valid.
Hence, conditions (3) and (4) cannot fail either.

If $j=0$, by the \wfc of the responder's dirty ledger, $(\tx_j,\stc_j) = (\epsilon, \left<\genesisstate\right>)$, so condition (7) cannot fail.

Let $r$ denote the round at which the bisection game was started.
If $j \geq 1$,
by the \wfc of the responder's dirty tree, for any consecutive pair of leaves at indices
$j-1$ and $j$, it holds that $(\tx_{j-1}, \tx_j) \mid \LOGdirty{r}{\prover} \preceq \ledger^{\cup}_r$ due to ledger safety.
As the \coracle is complete, by Definition~\ref{def:coracle-p}, it returns \emph{true} on $(\tx_{j-1},\tx_j)$,
implying that the condition (5) cannot fail.

Finally, by the \wfc of the responder's dirty tree, for any consecutive pair of leaves at indices $j-1$ and $j$, there exist a state $\st_{j-1}$ such that $\stc_{j-1} = \left<\st_{j-1}\right>$, and $\stc_j = \left<\transition(\st_{j-1},\tx_{j})\right>$.
As the \eoracle is complete, by Definition~\ref{def:eoracle-p}, $M(\st_{j-1},\tx_j)$ outputs a proof $\pi$ such that $\left<\transition\right>(\left<\st_{j-1}\right>, \tx_{j}, \pi) = \left<\transition(\st_{j-1},\tx_{j})\right>$.
Using the observations above, $\left<\transition\right>(\stc_{j-1}, \tx_{j}, \pi) = \stc_j$.
Consequently, condition (6) cannot fail.
Thus, the honest responder wins the bisection game against any adversary.
\end{proof}
\fi

Let $\textsc{Verify}(\pi, \mroot, \ell, i, v)$ be the verification function for Merkle proofs.
It takes a proof $\pi$, a Merkle root $\mroot$, the size of the tree $\ell$, an index for the leaf $0 \leq i < \ell$ and the leaf $v$ itself.
It outputs $1$ if $\pi$ is valid and $0$ otherwise.
The following proposition is a well-known folklore result about the security of Merkle trees, stating that it is impossible to prove proofs of inclusion for elements that were not present during the tree construction.
It extends the result that Merkle trees are
collision resistant~\cite{katz}.
\begin{proposition}[Merkle Security]
\label{prop:adv-merkle}
Let $H^s$ be a collision resistant hash function used in the binary Merkle trees.
For all PPT $\mathcal{A}$:
$
\Pr[(v,D,\pi,i) \gets \mathcal{A}(1^\lambda): \mroot = \textsc{MakeMT}(D).\mathrm{root} \land
D[i] \neq v \land \textsc{Verify}(\pi, \mroot, |D|, i, v) = 1] \leq \text{negl}(\lambda)
$.
\end{proposition}

\iflong
\begin{proof}
Suppose $\mathcal{A}$ is the adversary of the statement.
We will construct a hash collision adversary $\mathcal{A}'$ that calls $\mathcal{A}$ as a subroutine.
The adversary $\mathcal{A}'$ works as follows. It invokes $\mathcal{A}(1^\lambda)$, and obtains $v, D, \pi, i$.
Let $h^*_1, \ldots, h^*_{a-1}$ denote the hash values within $\pi$, where $a=\log|\ell|+1$ is the height of the Merkle tree.
Let $h_1, \ldots, h_{a-1}$ denote the inner nodes within the Merkle tree at the positions that correspond to those of $h^*_1, \ldots h^*_{a-1}$.
Let $\tilde{h}_1, \ldots, \tilde{h}_{a-1}$ denote the siblings of $h_1, \ldots, h_{a-1}$.
Define $\tilde{h}_{a} := \mroot$.
Then, $\tilde{h}_{1} = H(D[i])$, and for $i=1,\ldots,a-1$;
\begin{itemize}
    \item If $h_i$ is the left child of its parent, $\tilde{h}_{i+1} = H(h_i \concat \tilde{h}_i)$.
    \item If $h_i$ is the right child of its parent, $\tilde{h}_{i+1} = H(\tilde{h}_i \concat h_i)$.
\end{itemize}

Consider the event \textsc{Merkle-Collision} that $\mathcal{A}$ succeeds.
In that case, there exists a sequence of hash values $\tilde{h}^*_1, \ldots, \tilde{h}^*_{a}$ such that $\tilde{h}^*_1=H(v)$, $\tilde{h}^*_a=\mroot$, and for $i=1,\ldots,a-1$,
\begin{itemize}
    \item If $h_i$ is the left child of its parent, $\tilde{h}^*_{i+1} = H(h^*_i \concat \tilde{h}^*_i)$.
    \item If $h_i$ is the right child of its parent, $\tilde{h}^*_{i+1} = H(\tilde{h}^*_i \concat h^*_i)$.
\end{itemize}

Finally, for $i=1,\ldots,a$, define $h_{i,m}$ and $h_{i,c}$ as follows:
\begin{itemize}
    \item $h_{a,m}=\mroot$, $h_{a,c}=\mroot$.
    \item $h_{0,m}=v$, $h_{0,c}=D[i]$.
    \item If $h_i$ is the left child of its parent, $h_{i,m} = h^*_i \concat \tilde{h}^*_i$ and $h_{i,c} = h_i \concat \tilde{h}_i$.
    \item If $h_i$ is the right child of its parent, $h_{i,m} = \tilde{h}^*_i \concat h^*_i$ and $h_{i,m} = \tilde{h}_i \concat h_i$.
\end{itemize}

Finally, the adversary $\mathcal{A}'$ finds the first index $p$ for which there is a collision
\[
    H(h_{i,m}) = H(h_{i,c}) \text{ and } h_{i,m} \neq h_{i,c}
\]
and returns
$a:=h_{p,m}$ and $b:=h_{p,c}$, if such an index $p$ exists. Otherwise, it returns \textsc{Failure}.

In the case of \textsc{Merkle-Collision}, for $i=0,\ldots,a-1$, $h_{i+1,m} = H(h_{i,m})$, $h_{i+1,c} = H(h_{i,c})$.
As $v \neq D[i]$, a collision must have been found for at least one index $p \in [h-1]$.
Therefore, $\Pr[\mathcal{A}' \text{ succeeds}] = \Pr[\textsc{Merkle-Collision}]$.

However, since $\forall \text{ PPT } \mathcal{A}'$:
\[
\Pr[(a, b) \gets \mathcal{A}'(1^\lambda): a \neq b, H(a)=H(b)] \leq \text{negl}(\lambda)\,,
\]
therefore, $\Pr[\textsc{Merkle-Collision}] = \text{negl}(\lambda)$.
\end{proof}
\fi

\iflong
\begin{figure*}
    \centering
    \ifonecolumn
    \includegraphics[width=\linewidth]{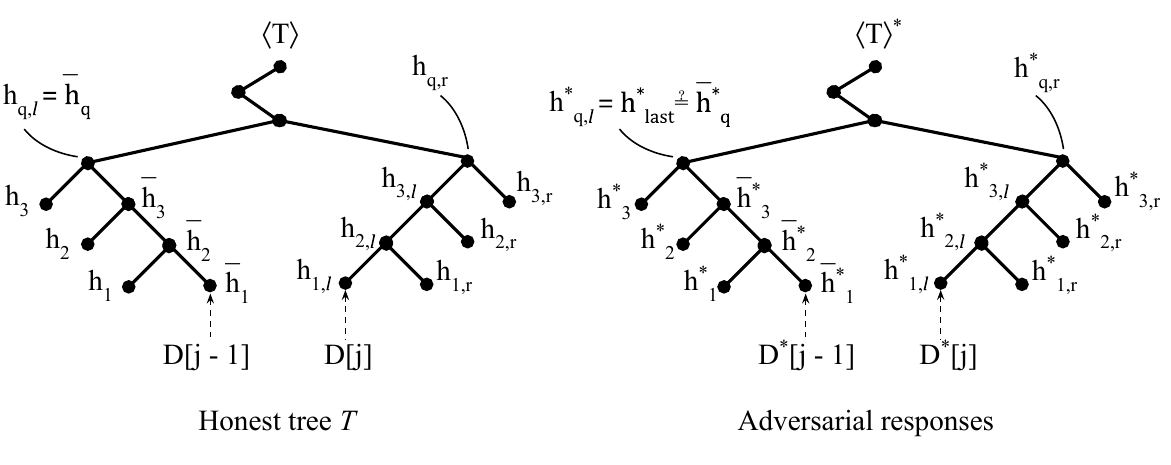}
    \fi
    \iftwocolumn
    \includegraphics[width=0.7\linewidth]{figures/bisection-soundness-proof.pdf}
    \fi
    \caption{The world in the view of the proof of Lemma~\ref{lem:bisection-game-entries}.
             Starred quantities (right-hand side) denote adversarially provided values.
             Unstarred quantities (left-hand side) denote the respective honestly provided values.
             The inner node at height $q$ from the leaves is the level containing the lowest
             common ancestor between leaves with indices $j$ and $j-1$.}
    \label{fig:bisection-soundness-proof}
\end{figure*}
\fi

The next lemma establishes an important result for our bisection game: That the
honest challenger can pinpoint the \emph{first point of disagreement} or \emph{last
point of agreement} indices $j - 1$ and $j$ within the responder's claimed tree. The
result stems from the fact that the data are organized into a Merkle tree which can
be explored, moving left or right, one level at a time, ensuring the invariant
that \emph{the first point of disagreement remains within the subtree explored}
at every step.

\begin{lemma}[Bisection Pinpointing]
\label{lem:bisection-game-entries}
Let $H^s$ be a collision resistant hash function.
Consider the following game among an honest challenger $\prover$, a verifier $\verifier$ and an adversarial responder $\prover^*$:
The challenger $\prover$ receives an array $D$ of size $\ell$ from $\prover^*$, and calculates the corresponding dirty tree $\dtreesp$ with root $\mroot$.
Then, $\prover$ plays the bisection game against $\prover^*$ claiming root $\mroot^* \neq \mroot$ and size $\ell$.
Finally, $\verifier$ outputs $(1, D^*[j-1], D^*[j])$ if $\prover$ wins the bisection game; otherwise, it outputs $(0, \bot, \bot)$.
Here, $D^*[j-1]$ and $D^*[j]$ are the two entries revealed by $\prover^*$ for the consecutive indices $j-1$ and $j$ during the bisection game. ($D^*[-1]$ is defined as $\bot$ if $j=0$.)
Then, for all PPT adversarial responder $\mathcal{A}$,
$\Pr[D \gets \mathcal{A}(1^\lambda);
     (1, D^*[j-1], D^*[j]) \gets \verifier
     \land (D^*[j-1] \neq D[j-1] \lor D^*[j] = D[j])]
     \leq \negl(\lambda)$.
\end{lemma}

\begin{proof}
Consider a PPT adversarial responder $\prover^*$ playing the bisection game against
the honest challenger $\prover$ at some round $r$.
Since the challenger is honest, his queries are valid and he does not time out.
For the responder to win the bisection game, she must satisfy all the conditions checked by Algorithm~\ref{alg.verifier}.

Consider the event $\textsc{Bad}$ that the responder wins.
Conditioned on $\textsc{Bad}$, the responder does not timeout and her replies are syntactically valid.
Let $a=\log\ell+1$ denote the height of the challenger's dirty tree.
At each round $i \in [a-1]$ of interactivity in the bisection game, the responder reveals two hash values $h^*_{a-i,l}$ and $h^*_{a-i,r}$.
The subscript $a-i$ signifies the alleged height of the nodes $h^*_{a-i,l}$ and $h^*_{a-i,r}$.
Let $h_{i,l}$ and $h_{i,r}$ denote the inner nodes in the honest challenger's dirty tree with the same positions as $h^*_{i,l}$ and $h^*_{i,r}$.
These will always exist, as the verifier limits the rounds of interaction to $a$.

At the first round, the responder reveals $h^*_{a-1,l}$ and $h^*_{a-1,r}$ as the alleged left and right children of its dirty tree root $\mroot^*$.
By condition (3) of Algorithm~\ref{alg.verifier}, $H(h^*_{a-1,l} \concat h^*_{a-1,r}) = \mroot^*$.
However, since $\mroot^* \neq \mroot = H(h_{a-1,l} \concat h_{a-1,r})$, either $h_{a-1,l} \neq h^*_{a-1,l}$ or $h_{a-1,r} \neq h^*_{a-1,r}$ or both.
Then, if $h_{a-1,l} \neq h^*_{a-1,l}$, the challenger picks $h^*_{a-1,l}$ to query next; else, he picks $h^*_{a-1,r}$.

We observe that if a node $h^*=h^*_{i,l}$ or $h^*=h^*_{i,r}$, $i \in \{2,\ldots,a-1\}$, returned by the responder, is queried by the honest challenger,
\begin{itemize}
    \item For the two children $h^*_{i-1,l}$ and $h^*_{i-1,r}$ of $h^*$, it holds that $h^* = H(h^*_{i-1,l} \concat h^*_{i-1,r})$ by condition (3).
    \item For the nodes $h$, $h_{i-1,l}$ and $h_{i-1,r}$ in the challenger's dirty tree that have the same positions as $h^*$, $h^*_{i-1,l}$ and $h^*_{i-1,r}$; $h = H(h_{i-1,l} \concat h_{i-1,r})$, and $h \neq h^*$.
    \item By implication, either $h_{i-1,l} \neq h^*_{i-1,l}$ or $h_{i-1,r} \neq h^*_{i-1,r}$ or both. If $h_{i-1,l} \neq h^*_{i-1,l}$, the challenger picks $h^*_{i-1,l}$ as its next query; else, it picks $h^*_{i-1,r}$ as its next query.
\end{itemize}
The queries continue until the challenger queries a node $h^*=h_{1,l}$ or $h^*=h_{1,r}$ returned by the responder, and the responder reveals the leaf $D^*[j]$ such that $H(D^*[j])=h^*$.
By induction, $h^*$ is different from the node $h = H(D[j])$ with the same position in the challenger's dirty tree. Thus, $D^*[j] \neq D[j]$.

If $j=0$, it must hold that $D^*[0] = (\epsilon, \left<\genesisstate\right>)$ by condition (7) of Algorithm~\ref{alg.verifier}.
However, since the challenger's dirty tree is \wf, $D[0] = (\epsilon, \left<\genesisstate\right>)$ as well.
Hence, $D^*[0] = D[0]$, therefore necessarily $j > 0$.

(When the provers hold MMRs instead of Merkle treees, responder's augmented dirty ledger entries $D^*[j-1]$ and $D^*[j]$ can lie in different Merkle trees held by the responder.
In this case, since the honest challenger did not initiate a bisection game between the responder's peak $\dtreesp^*_i$ containing $D^*[j-1]$ and his corresponding inner node $\dtreesp_i=\dtreesp^*_i$, $D^*[j-1]=D[j-1]$ with overwhelming probability.
To show this, we construct the PPT Merkle tree adversary that outputs $D^*[j-1]$, the honest challenger's leaves under $\dtreesp_i$, the responder's Merkle proof $\pi^*$ for $D^*[j-1]$ with respect to $\dtreesp_i=\dtreesp^*_i$ and the index of $D[j-1]$ within the subtree of $\dtreesp_i$, if $D^*[j-1] \neq D[j-1]$; and $\textsc{Failure}$ otherwise.
Since this adversary succeeds except with negligible probability in $\lambda$, $D^*[j-1]=D[j-1]$ with overwhelming probability, and this concludes the proof.
In the rest of this section, we assume that $j-1$ and $j$ lie in the same Merkle tree of the responder.)

As $j>0$,
there must exist a \emph{last} node queried by the challenger such that
for its children $h^*_{q,l}$ and $h^*_{q,r}$ revealed by the responder
at height $q$, it holds that
$h^*_{q,l}=h_{q,l}$ and $h^*_{q,r} \neq h_{q,r}$
(This is the last time the challenger went \emph{right}).
Define $h^*_{\textrm{last}}=h^*_{q,l}$
\iflong(see Figure~\ref{fig:bisection-soundness-proof})\fi.

By condition (4) of Algorithm~\ref{alg.verifier}, the Merkle proof $\pi^*$ for $D^*[j-1]$ is valid with respect to $\mroot^*$.
Let $h^*_1, h^*_2, \ldots, h^*_{a-1}$ denote the sequence of nodes on $\pi^*$
Let $\tilde{h}^*_1 := H(D^*[j-1])$
and define $\tilde{h}^*_{i+1}$, $i=1,\ldots,a-1$, recursively as follows:
$\tilde{h}^*_{i+1} := H(h^*_i \concat \tilde{h}^*_{i})$ if $h^*_i$ is the left child of its parent, and, $\tilde{h}^*_{i+1} := H(\tilde{h}^*_{i} \concat h^*_i)$ if $h^*_i$ is the right child of its parent.
Since $\pi^*$ is valid, $\tilde{h}^*_{a} = \mroot^*$ (The nodes $\tilde{h}^*_i$, $i \in [a-1]$, are the alleged nodes on the path connecting $D^*[j-1]$ to the root $\mroot^*$, and $h^*_i$ are their alleged siblings).

Let $h_i$, $i \in [a-1]$, denote the inner nodes in the challenger's dirty tree with the same positions as $h^*_i$.
Let $\tilde{h}_i$, $i \in [a-1]$, denote the inner nodes in the challenger's dirty tree on the path from $D[j-1]$ to $\mroot$.
These inner nodes satisfy the following relations for $i \in [a-1]$:
$\tilde{h}_a = \mroot$, $\tilde{h}_1 := H(D[j-1])$, $\tilde{h}_{i+1} = H(h_i \concat \tilde{h}_{i})$ if $h_i$ is the left child of its parent, and, $\tilde{h}_{i+1} = H(\tilde{h}_{i} \concat h_i)$ if $h_i$ is the right child of its parent.

Consider the event $\textsc{Discrepancy}$ that $\tilde{h}^*_q \neq h^*_{\textrm{last}}$ and the event \linebreak
$\textsc{Invalid-Proof}$ that $\tilde{h}^*_q = h^*_{\textrm{last}}\ \land\ D[j-1] \neq D^*[j-1]$.
Since
$\Pr[D[j-1] \neq D^*[j-1]\ |\ \textsc{Bad}] \leq
\Pr[\textsc{Discrepancy}] + \Pr[\textsc{Invalid-Proof}]$
we next bound the probabilities of these events.

We first construct a hash collision adversary $\mathcal{A}_1$ that calls the responder as a subroutine, and show that the event $\textsc{Discrepancy}$ implies that $\mathcal{A}_1$ succeeds.
For $i \in \{q,\ldots,a\}$, define $h^*_{i,c}$ as: $h^*_{a,c} := \mroot^*$ and $h^*_{i,c} := h^*_{i,l} \concat h^*_{i,r}$ if $i < h$.
Similarly, define $h^*_{i,m}$ as: $h^*_{a,m} := \mroot^*$, $h^*_{i,m} := \tilde{h}^*_{i} \concat h^*_{i}$ if $h_{i}$ is the right child of its parent, and $h^*_{i,m} := h^*_{i} \concat \tilde{h}^*_{i}$ if $h_{i}$ is the left child of its parent.

The adversary $\mathcal{A}_1$ calls the responder as a sub-routine, and obtains the values $h^*_{i,c}$ and $h^*_{i,m}$, $i \in \{q,\ldots,h\}$.
It finds the first index $p$ for which there is a collision
$
    H(h^*_{p,m}) = H(h^*_{p,c}) \text{ and } h^*_{p,m} \neq h^*_{p,c}
$
and returns
$a:=h^*_{p,m}$ and $b:=h^*_{p,c}$, if such an index $p$ exists.
Otherwise, it returns \textsc{Failure}.

In the case of $\textsc{Discrepancy}$, $\tilde{h}^*_q \neq h^*_{\textrm{last}} = h^*_{q,l}$.
Hence, it must be the case that $h^*_{q,m} \neq h^*_{q,c}$.
However, since $h^*_{a,m} = \mroot^* = h^*_{a,c}$, a collision must have been found for at least one index $i \in \{q,\ldots,a-1\}$.
Consequently, $\textsc{Discrepancy}$ implies that $\mathcal{A}_1$ succeeds.

We next construct a Merkle tree adversary $\mathcal{A}_2$ that calls the responder as a subroutine, and show that the event $\textsc{Invalid-Proof}$ implies that $\mathcal{A}_2$ succeeds.

Let $P$ denote the sequence of leaves in the challenger's dirty tree, \ie, within $D$, that lie under the subtree with root $h_{q,l}$.
Let $\pi$ denote the sub-sequence $h^*_1,\ldots, h^*_{q-1}$ within $\pi^*$.
The adversary $\mathcal{A}_2$ receives $P$ from the responder, and constructs a \wf dirty tree using $P$ in time $O(\text{poly}(\ell))$.
It then obtains the leaf $v := D^*[j-1]$ and the Merkle proof $\pi = (h^*_1,\ldots, h^*_{q-1})$ from the responder.
Finally, it returns $v$, $P$, $\pi$ and the index $\mathrm{idx}$ of the leaf $D[j-1]$ within the sequence $P$ such that $P[\mathrm{idx}]=D[j-1]$.

If $\textsc{Invalid-Proof}$, it must be the case that $\tilde{h}^*_q = h^*_{\textrm{last}} = h_{q,l}$ and $D[j-1] \neq D^*[j-1] = v$.
Hence, $\pi$ is a valid Merkle proof for $v$ with respect to the root $h_{q,l}$ of the Merkle tree with leaves $P$.
Moreover, $v \neq P[\mathrm{idx}]$.
Consequently, $\textsc{Invalid-Proof}$ implies than $\mathcal{A}_2$ succeeds.

Finally, by the fact that $H$ is a collision-resistant hash function and Lemma~\ref{prop:adv-merkle},
\begin{align*}
    &\Pr[D[j-1] \neq D^*[j-1]\ |\ \textsc{Bad}] \leq\\
    &\Pr[\textsc{Discrepancy}] + \Pr[\textsc{Invalid-Proof}] \leq\\
    &\Pr[\mathcal{A}_1 \text{ succeeds}] + \Pr[\mathcal{A}_2 \text{ succeeds}] \leq \text{negl}(\lambda)\,.
\end{align*}
Hence, for any PPT adversarial responder, the probability that the responder wins and $(D^*[j-1] \neq D[j-1]) \lor (D^*[j] = D[j])$ is negligible in $\lambda$.
\end{proof}

The next lemma ensures that an honest challenger can win in the bisection game
by leveraging sound consensus and execution oracles to resolve any disagreements at
the leaf level.

\begin{restatable}[Bisection Soundness]{lemma}{restateBisectionSoundness}
\label{lem:bisection-soundness}
Let $H^s$ be a collision resistant hash function.
Consider an execution that satisfies ledger safety and in which
the consensus and execution oracles are sound.
Then, for all PPT adversarial responders $\mathcal{A}$ claiming root $\mroot^*$ and size $\ell$, the honest challenger claiming $\mroot \neq \mroot^*$ and $\ell$
wins the bisection game against $\mathcal{A}$
with overwhelming probability in $\lambda$.
\end{restatable}

\begin{proof}
Consider an adversarial PPT responder $\prover^*$ playing against the honest challenger $\prover$ at some round $r$.
Since the challenger is honest, his queries are valid and he does not time out.
For the responder to win the bisection game, it must satisfy all the conditions checked by Algorithm~\ref{alg.verifier}.
Let $(\tx^*_{j-1},\stc^*_{j-1})$ and $(\tx^*_j,\stc^*_j)$ denote the two entries revealed by $\prover^*$ for the consecutive indices $j-1$ and $j$ in the event that it wins.

Define $\textsc{Consensus-Oracle}$ as the event that the responder wins and \linebreak $(\tx_{j-1},\stc_{j-1}) = (\tx^*_{j-1},\stc^*_{j-1})\land\tx^*_j \neq \tx_j$.
We construct a consensus oracle adversary $\mathcal{A}_1$ that calls $\prover^*$ as a subroutine and outputs $(\tx^*_{j-1},\tx^*_j,r)$.
By the \wfc of the challenger's dirty ledger and ledger safety, it holds that $(\tx_{j-1}, \tx_j) \mid \LOGdirty{r}{\prover} \preceq \ledger^{\cup}_r \preceq \ledger^{\cup}_{r+\nu}$.
Therefore, if $\textsc{Consensus-Oracle}$, it must be the case that $\tx^*_j \neq \tx_j$ does not immediately follow $\tx^*_{j-1}=\tx_{j-1}$ on $\ledger^{\cup}_{r+\nu}$ as every transaction on $\ledger^{\cup}_{r+\nu}$ is unique.
However, as the responder wins, the \coracle must have outputted \emph{true} on $(\tx^*_{j-1},\tx^*_j,r)$ by condition (5).
Hence, $\textsc{Consensus-Oracle}$ implies that $\mathcal{A}_1$ succeeds. %

Define $\textsc{Execution-Oracle}$ as the event that the responder wins and \linebreak $(\tx_{j-1},\stc_{j-1}) = (\tx^*_{j-1},\stc^*_{j-1})\ \land\ \tx^*_j = \tx_j\ \land\ \stc^*_j \neq \stc_j$.
By the \wfc of the challenger's dirty tree, there exist a state $\st_{j-1}$ such that $\stc_{j-1} = \left<\st_{j-1}\right>$, $\st_{j-1} = \transition^*(\genesisstate,\ledger[{:}j-1])$, and, $\stc_j = \left<\transition(\st_{j-1},\tx_{j})\right>$.
Therefore, if $\textsc{Execution-Oracle}$, it holds that $\stc_{j-1} = \stc^*_{j-1}$, and $\left<\transition(\st_{j-1},\tx_{j})\right> = \stc_j \neq \stc^*_j$.
However, as the responder wins, the \eoracle must have outputted \emph{true} on $\tx^*_j$, $\stc^*_{j-1}$ and $\stc^*_j$ by condition (6).
Thus, the responder must have given a proof $\pi$ such that $\left<\transition\right>(\stc^*_{j-1},\tx^*_j,\pi) = \stc^*_j$.
This implies $\left<\transition\right>(\stc_{j-1},\tx_j,\pi) = \stc^*_j \neq \left<\transition(\st_{j-1},\tx_{j})\right>$.

Finally, we construct an execution oracle adversary $\mathcal{A}_2$ that calls $\prover^*$ as a subroutine and receives $\pi$.
Then, using $\ledger$, $\mathcal{A}_2$ finds $\st_{j-1} = \transition^*(\genesisstate,\ledger[{:}j-1])$ in $O(\text{poly}(\ell))$ time.
It outputs $(\st_{j-1}, \tx_j, \pi)$.
Observe that if $\textsc{Execution-Oracle}$, then $\mathcal{A}_2$ succeeds.

Note that the event $(\tx_{j-1},\stc_{j-1}) = (\tx^*_{j-1},\stc^*_{j-1})\ \land\ (\tx_{j},\stc_{j}) \neq (\tx^*_{j},\stc^*_{j})\ \land\ \textit{Responder wins}$ is the union of the events $\textsc{Consensus-Oracle}$ and \linebreak $\textsc{Execution-Oracle}$:
\begin{align*}
    \Pr[&(\tx_{j-1},\stc_{j-1}) = (\tx^*_{j-1},\stc^*_{j-1}) \land\\
        &(\tx_{j},\stc_{j}) \neq (\tx^*_{j},\stc^*_{j}) \land\ \text{Responder wins}] =\\
    \Pr[&\textsc{Consensus-Oracle} \lor \textsc{Execution-Oracle}] \leq\\
    &\Pr[\mathcal{A}_1 \text{ succeeds}] + \Pr[\mathcal{A}_2 \text{ succeeds}] \leq \text{negl}(\lambda)\,.
\end{align*}
Moreover, by Lemma~\ref{lem:bisection-game-entries};
\begin{align*}
    \Pr[&((\tx_{j-1},\stc_{j-1}) \neq (\tx^*_{j-1},\stc^*_{j-1})\ \lor\\
        &(\tx_{j},\stc_{j}) = (\tx^*_{j},\stc^*_{j}))\ \land\ \text{Responder wins}] \leq \negl(\lambda)\,.
\end{align*}
Consequenty, $\Pr[\text{Responder wins}] = \negl({\lambda})$.
\end{proof}

\begin{theorem}[Succinctness]
\label{thm:succinctness}
Consider a consensus and execution oracle with $f$ and $g$ communication complexity respectively.
Then, the challenge game invoked at round $r$ with sizes $\ell_1$ and $\ell_2>\ell_1$ ends in $\log(\ell_1+\alpha(u+\nu))$ rounds of communication and
has a total communication complexity of $O(\log(\ell_1))+\alpha(u+\nu)(f(r)+g(r)))$.
\end{theorem}

\iflong
\begin{proof}
Suppose the challenge game was invoked on augmented dirty ledgers with (alleged) sizes $\ell_1$ and $\ell_2>\ell_1$ respectively.
The zooming phase of the challenge game does not require any communication among the provers and the verifier.

Suppose that at the end of the zooming phase, the provers play a bisection game on two Merkle trees with $\ell \leq \ell_1$ leaves.
By Lemma~\ref{lem:bisection-succinctness}, the bisection game ends in $\Theta(\log\ell)=\Theta(\log\ell_1)$ rounds and has a total communication complexity of $O(\log\ell + f(r)+g(r))=O(\log\ell_1+f(r)+g(r))$.

Suppose that the challenge game reaches the suffix monologue.
Since the verifier checks for at most $\alpha(u+\nu)$ extra entries, $(\tx_j,\stc_j)$, $j \in \{\ell_1,\ldots, \min(\ell_2,\ell_1+\alpha(u+\nu))\}$, %
at most $\alpha(u+\nu)$ entries are sent to the verifier by the challenger.
These entries have constant sizes since the transactions and the state commitments are assumed to have constant sizes.
Finally, the verifier can query the \coracle on the $\alpha(u+\nu)$ transaction pairs $(\tx_{j-1},\tx_j)$, $j \in \{\ell_1+1,\min(\ell_2,\ell_1+\alpha(u+\nu))\}$, and the \eoracle on the $\alpha(u+\nu)$ triplets $(\stc_{j-1},\tx_j,\stc_j)$, $j \in \{\ell_1+1,\min(\ell_2,\ell_1+\alpha(u+\nu))\}$, with $O(\alpha(u+\nu) f(r))$ and $O(\alpha(u+\nu) g(r))$ communication complexity respectively.
Hence, the total communication complexity of the challenge game becomes $O(\log\ell_1 + \alpha(u+\nu)(f(r)+g(r)))$.
\end{proof}
\fi
\ifshort
The proof follows from the Lemma~\ref{lem:bisection-succinctness}, and is provided in the full version of this paper.
\fi

By the Lipschitz property of the ledger, $|\ledger^{\cup}_r| < \alpha r$,
and $\alpha, \nu, u$ are constants. Superlight client constructions~\cite{nipopows,flyclient}
place $f$ in $\mathcal{O}(\text{poly} \log r)$, and $g$ is in $\mathcal{O}(\text{poly} \log r)$
if standard Merkle constructions~\cite{albassam2018fraud} are used and the transition function
$\transition$ ensures the state grows at most linearly, as is the case in all practical constructions.
In light of these quantities, the result of the above theorem establishes that
our protocol is also $\mathcal{O}(\text{poly} \log r)$ and, hence, \emph{succinct}.

\begin{theorem}[Completeness]
\label{thm:completeness}
Suppose the consensus and execution oracles are complete and the ledger is safe.
Then, the honest responder wins the challenge game against any PPT adversarial challenger.
\end{theorem}

\iflong
\begin{proof}
Suppose that at the end of the zooming phase, the challenger invoked the bisection game between one of the honest responder's peaks, $\left<\dtreesp\right>_i$, and a node $\left<\dtreesp\right>^*$ alleged to have the same position as $\left<\dtreesp\right>_i$ within the challenger's MMR.
By Lemma~\ref{lem:bisection-completeness}, the honest responder wins the bisection game.
If the challenger starts a suffix monologue instead of the bisection game at the end of the zooming phase, the responder automatically wins the challenge game.
Hence, the responder wins the challenge game.
\end{proof}

\ifshort
The proof follows from Lemma~\ref{lem:bisection-completeness}, and is provided in the full version of this paper.
\fi

\begin{proposition}
\label{prop:limited-divergence}
For any honest prover $\prover$ and round $r$,
$|\ledger^{\cup}_{r}| < |\ledger^{\prover}_r|+\alpha u$.
\end{proposition}

\begin{proof}
Towards contradiction, suppose $|\ledger^{\cup}_{r}| \geq |\ledger^{\prover}_r|+\alpha u$.
By ledger safety, there exists an honest prover $\prover'$ such that $\ledger^{\prover'}_{r}=\ledger^{\cup}_r$, which implies $|\ledger^{\prover'}_{r}| \geq |\ledger^{\prover}_r|+\alpha u$.
Again by ledger safety, $\ledger^{\prover}_r \preceq \ledger^{\prover'}_r$.
By ledger liveness, every transaction that is in $\ledger^{\prover'}_r$ and not in $\ledger^{\prover}_r$ becomes part of $\ledger^{\prover}_{r+u}$, for which $\ledger^{\prover}_{r} \preceq \ledger^{\prover}_{r+u}$ holds by ledger safety.
Hence, $\ledger^{\prover}_{r} \preceq \ledger^{\prover'}_{r} \preceq \ledger^{\prover}_{r+u}$
and, $|\ledger^{\prover}_{r+u}| \geq |\ledger^{\prover'}_{r}| \geq |\ledger^{\prover}_r|+\alpha u$.
However, this is a violation of the ledger Lipschitz property.
Consequently, it should be the case that $|\ledger^{\cup}_{r}| < |\ledger^{\prover}_r|+\alpha u$.
\end{proof}

\begin{lemma}[Monologue Succinctness]
\label{lem:monologue-succinctness}
    Consider an execution of a consensus protocol
    which is Lipschitz with parameter $\alpha$ and
    has liveness with parameter $u$.
    Consider the challenge game instantiated
    with a collision-resistant hash function $H^s$
    and a consensus oracle which is sound with parameter $\nu$.
    For all PPT adversarial challengers $\mathcal{A}$, if the game administered by the honest verifier
    among $\mathcal{A}$ and
    the honest responder $\prover$ at round $r$ reaches the suffix monologue,
    the adversary cannot reveal $\alpha(u +\nu)$ or more
    entries and win the game except with negligible probability.
\end{lemma}
\begin{proof}
    Suppose the game between the challenger $\mathcal{A}$ and the honest responder $\prover$ reaches the suffix monologue.
    Consider the event $\textsc{Bad}$ that the challenger reveals $\beta\geq\alpha(u +\nu)$ entries and wins the game.
    Let $D = ((\tx_1, \stc_1),\linebreak (\tx_2, \stc_2), \ldots, (\tx_\beta, \stc_\beta))$ denote these entries, and $(\tx_0, \stc_0)$ the responder's last entry prior to the monologue phase.
    Because $\prover$ is in agreement with $\tx_0$, therefore
    $\tx_0 = \ledger^{\prover}_r[-1]$.
    Let $J = (\tx_0, \tx_1, \ldots, \tx_\beta)$.
    Since the challenger wins, the verifier has invoked the consensus oracle
    $\alpha(u +\nu)$ times
    for all consecutive pairs of transactions within $K = J[{:}\alpha(u +\nu)]$.
    At each invocation, the consensus oracle has returned \emph{true}.

    We next construct a consensus oracle adversary $\mathcal{A}'$ that calls $\mathcal{A}$ as a subroutine.
    If $\beta \geq \alpha(u +\nu)$, $\mathcal{A}'$ identifies the first index $p \in [\alpha(u+\nu)]$ such that $\tx_p$ does not immediately follow $\tx_{p-1}$ on $\ledger^{\cup}_{r +\nu}$, and outputs $(\tx_{p-1},\tx_{p},r)$.
    If $\beta < \alpha(u +\nu)$, $\mathcal{A}'$ outputs $\textsc{Failure}$.

    By the ledger Lipschitz property, $|\ledger^{\prover}_{r+\nu}| < |\ledger^{\prover}_{r}|+\alpha \nu$.
    Moreover, by Lemma~\ref{prop:limited-divergence}, $|\ledger^{\cup}_{r +\nu}| < |\ledger^{\prover}_{r+\nu}| + \alpha u$.
    Thus, $|\ledger^{\cup}_{r +\nu}| < |\ledger^{\prover}_{r}| + \alpha(u+\nu)$.

    Let $\ell = |\ledger^{\prover}_r|$.
    By ledger safety, $\tx_0 = \ledger^{\prover}_r[\ell-1]=\ledger^{\cup}_{r +\nu}[\ell-1]$.
    Hence, if $(\tx_0, \tx_1) \mid \ledger^{\cup}_{r +\nu}$, $\tx_1=\ledger^{\cup}_{r +\nu}[\ell]$ as every transaction on $\ledger^{\cup}_{r +\nu}$ is unique.
    By induction, either there exists an index $i \in [\beta]$ such that $\tx_i$ does not immediately follow $\tx_{i-1}$ on $\ledger^{\cup}_{r +\nu}$, or $|\ledger^{\cup}_{r +\nu}| \geq \ell + \beta$ and $\tx_i = \ledger^{\cup}_{r +\nu}[\ell+i-1]$ for all $i \in [\beta]$.

    Finally, if $\beta \geq \alpha(u+\nu)$, there exists an index $i \in [\alpha(u+\nu)]$ such that $\tx_i$ does not immediately follow $\tx_{i-1}$ on $\ledger^{\cup}_{r +\nu}$.
    Thus, $\Pr[\textsc{Bad}] = \Pr[\mathcal{A}' \text{succeeds}]$. However, by the soundness of the consensus oracle, $\forall$ PPT $\mathcal{A}'$, $\Pr[\mathcal{A}' \text{succeeds}] = \negl(\lambda)$.
    Therefore, $\Pr[\textsc{Bad}] = \negl(\lambda)$.
\end{proof}

\begin{theorem}[Soundness]
\label{thm:soundness}
Let $H^s$ be a (keyed) collision resistant hash function.
Suppose the consensus and execution oracles are complete and sound.
Then, for all PPT adversarial responders $\mathcal{A}$, an honest challenger wins the challenge game against $\mathcal{A}$ with overwhelming probability in $\lambda$.
\end{theorem}

\iflong
\begin{proof}
Suppose that at the end of the zooming phase, the honest challenger $\prover$ identified one of the responder $\prover^*$'s peaks, $\left<\dtreesp\right>^*_i$, as being different from a node $\left<\dtreesp\right>$ within the challenger's MMR that has the same position as $\left<\dtreesp\right>^*_i$.
In this case, the challenger initiates a bisection game between $\left<\dtreesp\right>^*_i$ and $\left<\dtreesp\right>$.
By Lemma~\ref{lem:bisection-soundness}, the honest challenger wins the bisection game with overwhelming probability.

Suppose the challenger observes that the peaks shared by the responder correspond to the peaks of a \wf MMR. Then, at the end of the zooming phase, the honest challenger starts the suffix monologue.
Let $\ell$ and $\ell^*$ denote the challenger's and the responder's (alleged) augmented dirty ledger sizes respectively.
Let $r$ denote the round at which the challenge game was started.
During the suffix monologue, the challenger reveals its augmented dirty ledger entries $(\tx_j,\stc_j)$ at the indices $\ell^*,\ldots, \min{(\ell,\ell^*+\alpha(u+\nu))}-1$.
Then, for all $j \in \{\ell^*+1,\ell^*+2,\ldots,$ $\min{(\ell,\ell^*+\alpha(u+\nu))}-1\}$, the verifier checks the transactions and the state transitions between $(\tx_{j-1},\stc_{j-1})$ and $(\tx_j,\stc_j)$.
The verifier does the same check between the responder's last (alleged) augmented dirty ledger entry $(\tx^*_{\ell^*-1},\stc^*_{\ell^*-1})$ and $(\tx_{\ell^*},\stc_{\ell^*})$.

Consider the event $\textsc{Equal}$ that $(\tx^*_{\ell^*-1},\stc^*_{\ell^*-1}) = (\tx_{\ell^*-1},\stc_{\ell^*-1})$.
By the \wfc of the challenger's augmented dirty ledger, for any pair of leaves at indices $j-1$ and $j$, $j \in \{\ell^*,\ldots, \min{(\ell,\ell^*+\alpha(u+\nu))}-1\}$, it holds that $(\tx_{j-1}, \tx_j) \mid \LOGdirty{r}{\prover}$, thus, on $\ledger^{\cup}_r$ by ledger safety.
As the \coracle is complete, by Definition~\ref{def:coracle-p}, it returns \emph{true} on all $(\tx_{j-1},\tx_j)$ for $j \in \{\ell^*,\ldots,\min{(\ell,\ell^*\alpha(u+\nu))}-1\}$.
Similarly, by the \wfc of the challenger's augmented dirty ledger, for any pair of leaves at indices $j-1$ and $j$, $j \in \{\ell^*,\ldots, \min{(\ell,\ell^*+\alpha(u+\nu))}-1\}$, there exists a state $\st_{j-1}$ such that $\stc_{j-1} = \left<\st_{j-1}\right>$, and, $\stc_j = \left<\transition(\st_{j-1},\tx_{j})\right>$.
As the \eoracle is complete, by Definition~\ref{def:eoracle-p}, for all $j \in \{\ell^*,\ldots, \min{(\ell,\ell^*+\alpha(u+\nu))}-1\}$, $M(\st_{j-1},\tx_j)$ outputs a proof $\pi_j$ such that $\left<\transition\right>(\left<\st_{j-1}\right>, \tx_{j}, \pi_j) = \left<\transition(\st_{j-1},\tx_{j})\right>$.
Thus, for all $j \in \{\ell^*,\ldots,$
$\min{(\ell,\ell^*+\alpha(u+\nu))}-1\}$, the verifier obtains a proof $\pi_j$ such that $\left<\transition\right>(\stc_{j-1}, \tx_{j}, \pi_j) = \stc_j$.
In other words, if the challenge protocol reaches the suffix monologue and $\textsc{Equal}$, the honest challenger wins the suffix monologue.

Finally, consider the Merkle tree adversary $\mathcal{A}'$ that calls the responder $\prover^*$ as a subroutine.
Let $\pi^*$ denote the Merkle proof revealed by the responder for $(\tx^*_{\ell^*-1},\stc^*_{\ell^*-1})$ with respect to its last (alleged) peak $\left<\dtreesp\right>^*$.
Let $\left<\dtreesp\right>$ denote the corresponding node in the challenger's MMR.
Let $D$ denote the sequence of augmented dirty ledger entries held by the honest challenger in the subtree rooted at $\left<\dtreesp\right>$.
Let $\mathrm{idx} := |D|$ denote the size of this subtree.
If the game reaches the suffix monologue and $\neg\textsc{Equal}$, $\mathcal{A}'$ returns $v := (\tx^*_{\ell^*-1},\stc^*_{\ell^*-1})$, $D$, $\pi$ and $\mathrm{idx}$.
Otherwise, it returns $\textsc{Failure}$.

If the game reaches the suffix monologue, $\left<\dtreesp\right>^* = \left<\dtreesp\right>$, and $\pi$ is valid with respect to $\left<\dtreesp\right>$.
Then, if $(\tx^*_{\ell^*-1},\stc^*_{\ell^*-1}) \neq (\tx_{\ell^*-1},\stc_{\ell^*-1})$, therefore $\left<\dtreesp\right> = \textsc{MakeMT}(D).\textrm{root}$, $D[\mathrm{idx}] \neq v$, and $\textsc{Verify}(\pi, \left<\dtreesp\right>, \mathrm{idx}, v) = 1$.
Conditioned on the fact that the challenge game reaches the suffix monologue, by Proposition~\ref{prop:adv-merkle}, $\Pr[\neg\textsc{Equal}] = \Pr[\mathcal{A}' \text{ succeeds}] = \negl(\lambda)$.
Thus, the honest challenger wins the challenge protocol with overwhelming probability.
\end{proof}
\else
The proof of this theorem follows from Lemma~\ref{lem:bisection-soundness}, and is provided in the full version of the paper.
\fi

\begin{theorem}[Tournament Runtime]
\label{thm:tournament-runtime}
Suppose the consensus and execution oracles are complete and sound, and have $f$ and $g$ communication complexity respectively.
Consider a tournament started at round $r$ with $n$ provers.
Given at least one honest prover, for any PPT adversary $\mathcal{A}$, the tournament ends in $2n\log(|\ledger^{\cup}_r|+\alpha(u+\nu))$ rounds of communication and has a total communication complexity of $O(2n\log(|\ledger^{\cup}_r|+\alpha(u+\nu))+2n\alpha(u+\nu)(f(r)+g(r)))$, with overwhelming probability in $\lambda$.
\end{theorem}

\begin{proof}
By the end of the first step, size of the set $\mathcal{S}$ can be at most $2$.
Afterwards, each step of the tournament adds at most one prover to $\mathcal{S}$ and the number of steps is $n-1$.
Moreover, at each step, either there is exactly one challenge game played, or if $k>1$ games are played, at least $k-1$ provers are removed from $\mathcal{S}$.
Hence, the maximum number of challenge games that can be played over the tournament is at most $2n-1$.

Recall that the size alleged by $\prover_i$ is at most the size alleged by $\prover_{i+1}$, $i \in [n-1]$.
Let $i^*$ be the first round where an honest prover plays the challenge game.
If $i^*>1$, until round $i^*$, the sizes alleged by the provers are upper bounded by $|\ledger^{\cup}_r|$.
From round $i^*$ onward, at each round, the prover $\largest$ claiming the largest size is either honest or must have at least once won the challenge game as a \emph{challenger} against an honest responder.
During the game against the honest responder,
by Lemma~\ref{lem:monologue-succinctness}, $\largest$ could not have revealed $\alpha(u+\nu)$ or more entries except with negligible probability.
Hence, from round $i^*$ onward, with overwhelming probability, the size claimed by $\largest$ at any round can at most be $|\ledger^{\cup}_r|+\alpha(u+\nu)-1$.
Thus, with overwhelming probability, by Theorem~\ref{thm:succinctness}, each challenge game ends after at most $\log(|\ledger^{\cup}_r|+\alpha(u+\nu))$ rounds of interactivity and has total communication complexity $O(\log(|\ledger^{\cup}_r|+\alpha(u+\nu))+\alpha(u+\nu)(f(r)+g(r)))$.
Consequently, with overwhelming probability, the tournament started at round $r$ with $n$ provers ends in at most $2n\log(|\ledger^{\cup}_r|+\alpha(u+\nu))$ rounds of interactivity and has total communication complexity $O(2n\log(|\ledger^{\cup}_r|+\alpha(u+\nu))+2n\alpha(u+\nu)(f(r)+g(r)))$.
\end{proof}

\begin{lemma}
\label{lem:challenge-security}
Consider a challenge game invoked by the verifier at some round $r$.
If at least one of the provers $\prover$ is honest, for all PPT adversarial $\mathcal{A}$, the state commitment obtained by the verifier at the end of the game between $\prover$ and $\mathcal{A}$ satisfies state security with overwhelming probability.
\end{lemma}

\iflong
\begin{proof}
If the challenger is honest, by Theorem~\ref{thm:soundness}, he wins the challenge game with overwhelming probability
and the verifier accepts his state commitment.

Suppose the responder $\prover$ of the challenge game is honest, and it is challenged by a challenger $\prover^*$.
If $\prover^*$ starts a bisection game, by Lemma~\ref{lem:bisection-completeness}, $\prover^*$ loses the challenge game and $\prover$ wins the game.
In this case, the verifier accepts the state commitment given by the honest responder.
On the other hand, if the challenge game reaches the suffix monologue and the challenger loses the monologue, the verifier again accepts the state commitment given by the honest responder.
As the state commitment of an honest prover satisfies security as given by Definition~\ref{def:state-security}, in all of the cases above, the commitment accepted by the verifier satisfies state security with overwhelming probability.

Finally, consider the event $\textsc{Win}$ that the game reaches the suffix monologue and the challenger wins.
Let $\ell$ and $\ell^*$ denote the responder's and the challenger's (alleged) augmented dirty ledger lengths respectively.
During the suffix monologue, the challenger reveals its alleged entries $(\tx^*_i,\stc^*_i)$ at indices $i = \ell+1,\ldots, \min{(\ell^*,\ell+\alpha(u+\nu))}-1$.
Let $(\tx_{\ell-1},\stc_{\ell-1})$ denote the responder's last entry.
As the challenger wins, \coracle must have returned true on $(\tx_{\ell-1},\tx^*_{\ell})$ and $(\tx^*_{i-1},\tx^*_i)$ for all $i \in \{\ell+1,\ldots,\min{(\ell^*,\ell+\alpha(u+\nu))}-1\}$.
Similarly, for all $i \in \{\ell,\ldots,\min{(\ell^*,\ell+\alpha(u+\nu))}\}$, \eoracle must have outputted a proof $\pi_{i-\ell+1}$ such that
$\left<\transition\right>(\stc_{\ell-1}, \tx^*_{\ell}, \pi_1) = \stc^*_{\ell}$ and
it holds $\left<\transition\right>(\stc^*_{i-1}, \tx^*_i, \pi_{i-\ell+1}) = \stc^*_i$ for $i \in \{\ell+1,\ldots,\min{(\ell^*,\ell+\alpha(u+\nu))}-1\}$.

Let $D$ denote the sequence $\tx_{\ell-1}, \tx^*_{\ell}, \ldots,$ $\tx^*_{\min{(\ell^*,\ell+\alpha(u+\nu))}-1}$ of transactions.
Consider the event $\textsc{Consensus-Oracle}$ that $\textsc{Win}$ holds, $\ell^* < \ell + \alpha(u+\nu)$, and there exists an index $i \in \{1,\ldots,\ell^*-\ell\}$ such that $D[i]$ does not immediately follow $D[i-1]$ on $\ledger^{\cup}_{r+\nu}$.
We next construct a consensus oracle adversary $\mathcal{A}_1$ that calls $\prover^*$ as a subroutine.
The adversary $\mathcal{A}_1$ identifies the first index $p>0$ such that $D[p]$ does not immediately follow $D[p-1]$ on $\ledger^{\cup}_{r+\nu}$ if such an index exists, and outputs $(D[p-1],D[p],r)$.
Otherwise, $\mathcal{A}_1$ outputs $\textsc{Failure}$.
Hence, $\textsc{Consensus-Oracle}$ implies that $\mathcal{A}_1$ succeeds.

Let $S$ denote the sequence $\stc_{\ell-1}, \stc^*_{\ell}, \ldots,$ $\stc^*_{\min{(\ell^*,\ell+\alpha(u+\nu))}-1}$.
Define $\st_i=\transition^*(\genesisstate, \LOGdirty{r}{\prover} || (\tx^*_{\ell}, \ldots, \tx^*_{\ell+i-1}))$ for $i \in \{1,\ldots,\ell^*-\ell\}$ ($\st_0 = \transition^*(\genesisstate, \LOGdirty{r}{\prover})$).
Consider the event $\textsc{Execution-Oracle}$ that $\textsc{Win}$ holds, $\ell^* < \ell + \alpha(u+\nu)$, $\neg\textsc{Consensus-Oracle}$ holds, and $S[i] \neq \left<\st_{i}\right>$ for at least one index $i \in \{1,\ldots,\ell^*-\ell\}$.
We next construct an execution oracle adversary $\mathcal{A}_2$ that calls $\prover^*$ as a subroutine.
Using $\LOGdirty{r}{\prover}$, $\mathcal{A}_2$ finds $\st_i$ for all $i \in \{0,1,\ldots,\ell^*-\ell\}$ in $O(\text{poly}(|\LOGdirty{r}{\prover}|)$ time.
Then, $\mathcal{A}_2$ identifies the first index $p>0$ such that $S[p] \neq \left<\st_p\right>$ if such an index exists, and outputs $\st = \st_{p-1}$, $\tx = D[p] = \tx^*_{\ell+p-1}$, and $\pi=\pi_p$.
Otherwise, $\mathcal{A}_2$ outputs $\textsc{Failure}$.
Since $\left<\transition\right>(S[i-1], D[i], \pi_i) = S[i]$ for $i \in \{0,1,\ldots,\min{(\ell^*-\ell,\alpha(u+\nu))}-1\}$, the $\textsc{Execution-Oracle}$ implies that 
\[\left<\transition(\st_{p-1},D[p])\right> = \left<\st_p\right> \neq S[p] = \left<\transition\right>(S[p-1],D[p],\pi_p) = \left<\transition\right>(\left<\st_{p-1}\right>,D[p],\pi_p),\]
\ie,
$\mathcal{A}_2$ succeeds.

Finally, if $\textsc{Win}\ \land$ $\neg\textsc{Consensus-Oracle}\ \land$ $\neg\textsc{Execution-Oracle}\ \land$ $\ell^* < \ell+\alpha(u+\nu)$, the verifier accepts the commitment $\stc^*_{\ell^*-1}$, which satisfies state security by Definition~\ref{def:state-security} (Here, $\ledger = \LOGdirty{r}{\prover}\concat (\tx^*_{\ell},\ldots,\tx^*_{\ell^*-1}) \preceq \ledger^{\cup}_{r+\nu}$ and $\stc^*_{\ell^*-1} = \left<\transition^*(\st_0, \ledger)\right> = \stc$).
However,
\begin{align*}
    \Pr[&\textsc{Consensus-Oracle} \lor \textsc{Execution-Oracle}] \leq\\
    &\Pr[\mathcal{A}_1 \text{ succeeds}] + \Pr[\mathcal{A}_2 \text{ succeeds}] \leq \text{negl}(\lambda)\,.
\end{align*}
Moreover, by Lemma~\ref{lem:monologue-succinctness}, $\prover^*$ cannot reveal $\alpha(u+\nu)$ or more entries and win the game except with negligible probability.
Hence,
\begin{align*}
    \Pr[\textsc{Win}] = &\negl(\lambda) + \Pr[\textsc{Win} \land \neg\textsc{Consensus-Oracle}\\
    &\land \neg\textsc{Execution-Oracle} \land \ell^* < \ell+\alpha(u+\nu)]\,.
\end{align*}
which implies that either $\Pr[\textsc{Win}]=\negl(\lambda)$ or conditioned on $\textsc{Win}$, the commitment accepted by the verifier satisfies state security except with negligible probability.
Consequently, in a challenge game invoked by the verifier at some round $r$, if at least one of the provers is honest, the state commitment obtained by the verifier satisfies state security except with negligible probability.
\end{proof}
\else
The proof follows from Theorem~\ref{thm:soundness} and Lemmas~\ref{lem:bisection-completeness} and \ref{lem:monologue-succinctness}
and is provided in the full version.
\fi

\begin{theorem}[Security]
\label{thm:security}
Suppose the consensus and execution oracles are complete and sound, and have $f$ and $g$ communication complexity respectively.
Consider a tournament started at round $r$ with $n$ provers.
Given at least one honest prover, for any PPT adversary $\mathcal{A}$, the state commitment obtained by the prover at the end of the tournament satisfies State Security with overwhelming probability in $\lambda$.
\end{theorem}

\iflong
\begin{proof}
Let $\prover_{i^*}$ denote an honest prover within $\mathcal{P}$.
Let $n = |\mathcal{P}|-1$ denote the total number of rounds.
By Theorems~\ref{thm:completeness} and~\ref{thm:soundness}, $\prover_{i^*}$ wins every challenge game and stays in $\mathcal{S}$ after step $i^*$ with overwhelming probability.

The prover $\largest$ with the largest alleged MMR at the end of each step $i \geq i^*$ is either $\prover_{i^*}$ or has a larger (alleged) MMR than the one held by $\prover_{i^*}$.
In the first case, as $\prover_{i^*}$ is honest, its state commitment satisfies safety and liveness per Definition~\ref{def:state-security}.
In the latter case, $\largest$ must have played the challenge game with $\prover_{i^*}$.
Then, by Lemma~\ref{lem:challenge-security}, the state commitment of $\largest$ satisfies safety and liveness per Definition~\ref{def:state-security} with overwhelming probability.
Consequently, the state commitment obtained by the verifier at the end of the tournament, \ie, the commitment of $\largest$ at the end of round $n \geq i^*$, satisfies safety and liveness with overwhelming probability.
\end{proof}
\fi
\ifshort
The proof follows from Lemma~\ref{lem:challenge-security}, and is provided in the full
version of this paper.
\fi

\begin{theorem}[Prover Complexity]
\label{thm:prover-complexity}
When updating the MMR on a rolling basis, provers do constant amortized number of hash computations per transaction.
Moreover, a node with an MMR of $\ell$ leaves can append a new leaf to its MMR with at most $O(\log{\ell})$ hash computations.
\end{theorem}

\begin{proof}
Given an augmented dirty ledger of length $\ell$, a prover can construct the corresponding MMR with $O(\ell)$ operations upon entering the challenge game.
This is because each prover can obtain the binary representation of $\ell$ with $O(\ell)$ operations, and create each of the $k$ Merkle trees $\dtreesp_i$, $i \in [k]$, with $O(2^{q_i})$ hash computations, making the total compute complexity $O(\ell)$ (\cf Section~\ref{sec:constructing-dirty-ledgers}).
Hence, updating the MMR on a rolling basis, each prover can obtain an MMR with $\ell$ leaves with $O(1)$ amortized number of operations per transaction.

Finally, a node with an MMR of $\ell$ leaves can append a new leaf to its MMR with $O(\log{\ell})$ hash computations; since in the worst case, it only needs to combine the existing $\log\ell$ hashes to update the MMR.
Hence, per each new transaction, each prover only incurs at most logarithmic compute complexity.
\end{proof}